\documentclass[a4paper,USenglish]{lipics-v2016-arxiv}
 
\usepackage{microtype}
\usepackage{slashbox}
\usepackage{xspace}
\usepackage{amssymb,amsmath}

\title{Lower Bounds for Protrusion Replacement by Counting Equivalence Classes\footnote{This work was supported by NWO Veni grant ``Frontiers in Parameterized Preprocessing'' and NWO
Gravitation grant ``Networks''. An extended abstract of this work appeared in the proceedings of the 11th International Symposium on Parameterized and Exact Computation (IPEC 2016).}}

\author[1]{Bart M.\,P. Jansen}
\author[2]{Jules J.\,H.\,M. Wulms}
\affil[1]{Eindhoven University of Technology\\
  P.O. Box 513, Eindhoven, The Netherlands\\
  \texttt{b.m.p.jansen@tue.nl}}
\affil[2]{Eindhoven University of Technology\\
  P.O. Box 513, Eindhoven, The Netherlands\\
 \texttt{j.j.h.m.wulms@tue.nl}}
\authorrunning{B.\,M.\,P. Jansen and J.\,J.\,H.\,M.\, Wulms} 

\Copyright{Bart M.\,P. Jansen and Jules J.\,H.\,M.\, Wulms}

\subjclass{G.2.1 Combinatorics, G.2.2 Graph Theory}
\keywords{protrusions, boundaried graphs, independent set, equivalence classes, finite integer index}

\EventEditors{}
\EventNoEds{2}
\EventLongTitle{}
\EventShortTitle{}
\EventAcronym{}
\EventYear{2016}
\EventDate{}
\EventLocation{}
\EventLogo{}
\SeriesVolume{}
\ArticleNo{1}

\theoremstyle{plain}
\newtheorem{claim}{Claim}
\newtheorem{proposition}[theorem]{Proposition}
\newtheorem{observation}[theorem]{Observation}
\newcommand{\IndependentSet}{\textsc{Independent Set}\xspace}
\newcommand{\VertexCover}{\textsc{Vertex Cover}\xspace}
\newcommand{\DominatingSet}{\textsc{Dominating Set}\xspace}
\newcommand{\IS}{\ensuremath{\mathrm{\textsc{is}}}\xspace}
\newcommand{\VC}{\ensuremath{\mathrm{\textsc{vc}}}\xspace}
\newcommand{\DS}{\ensuremath{\mathrm{\textsc{ds}}}\xspace}
\newcommand{\s}{\ensuremath{\mathfrak{s}}\xspace}
\newcommand{\Oh}{\ensuremath{\mathcal{O}}\xspace}
\newcommand{\tw}{\ensuremath{\textbf{tw}}\xspace}
\newcommand{\pw}{\ensuremath{\textbf{pw}}\xspace}

\newcommand{\ett}{\ensuremath{\triangle}\xspace}
\newcommand{\optis}{\ensuremath{\textsc{opt}_{\textsc{\IS}}}\xspace}
\newcommand{\optvc}{\ensuremath{\textsc{opt}_{\textsc{\VC}}}\xspace}
\newcommand{\optds}{\ensuremath{\textsc{opt}_{\textsc{\DS}}}\xspace}
\newcommand{\opt}{\ensuremath{\textsc{opt}}\xspace}

\newcommand{\bigmid}{\,\big\vert\,}
\newcommand{\vstart}{\ensuremath{v_{\mathrm{start}}}\xspace}
\newcommand{\vend}{\ensuremath{v_{\mathrm{end}}}\xspace}
\newenvironment{claimproof}[1][\proofname]{\begin{proof}[#1]\renewcommand{\qedsymbol}{\claimqed}}{\end{proof}\renewcommand{\qedsymbol}{\plainqed}}

\let\plainqed\qedsymbol

\begin{document}

\maketitle

\begin{abstract}
Garnero et al.~[SIAM J.\, Discrete Math.\, 2015, 29(4):1864--1894] recently introduced a framework based on dynamic programming to make applications of the \emph{protrusion replacement} technique constructive and to obtain explicit upper bounds on the involved constants. They show that for several graph problems, for every boundary size~$t$ one can find an explicit set~$\mathcal{R}_t$ of \emph{representatives}. Any subgraph~$H$ with a boundary of size~$t$ can be replaced with a representative~$H' \in \mathcal{R}_t$ such that the effect of this replacement on the optimum can be deduced from~$H$ and~$H'$ alone. Their upper bounds on the size of the graphs in~$\mathcal{R}_t$ grow triple-exponentially with~$t$. In this paper we complement their results by lower bounds on the sizes of representatives, in terms of the boundary size~$t$. For example, we show that each set of planar representatives~$\mathcal{R}_t$ for \textsc{Independent Set} or \textsc{Dominating Set} contains a graph with~$\Omega(2^t / \sqrt{4t})$ vertices. This lower bound even holds for sets that only represent the planar subgraphs of bounded pathwidth. To obtain our results we provide a lower bound on the number of equivalence classes of the canonical equivalence relation for \IndependentSet on $t$-boundaried graphs. We also find an elegant characterization of the number of equivalence classes in general graphs, in terms of the number of monotone functions of a certain kind. Our results show that the number of equivalence classes is at most~$2^{2^t}$, improving on earlier bounds of the form~$(t+1)^{2^t}$.
\end{abstract}

\section{Introduction}

Protrusion replacement is a versatile tool for attacking optimization problems on graphs. When applied to solve an optimization problem on a graph~$G$, the main idea is the following: repeatedly replace a \emph{protrusion} subgraph~$H \subseteq G$ that interacts with the rest of~$G$ through a small boundary, by a smaller \emph{representative} subgraph~$H'$. Suppose that we can ensure that (i) the change~$\Delta$ in the optimum caused by this replacement only depends on~$H$ and~$H'$, and that (ii)~we can efficiently analyze~$H$ to find a suitable replacement~$H'$ and the corresponding~$\Delta$. Then we can solve the problem on~$G$ by solving it on the smaller graph and adding~$\Delta$ to the final result. In recent years, protrusion replacement has been applied to obtain approximation algorithms~\cite{FominLMPS16,FominLMS12}, kernelization algorithms~\cite{BodlaenderFLPST09,FominLMPS16,FominLMS12,GajarskyHOORRVS13,KimLPRRSS16}, and fixed-parameter tractable algorithms~\cite{FominLMS12,KimLPRRSS16}. The generality of protrusion replacement comes at a price: it often results in proofs that efficient algorithms of a certain type \emph{exist}, without showing explicitly how such algorithms can be \emph{constructed} and without giving any explicit bounds on the constant factors involved in the analysis. This non-constructivity stems from the use of a property called \emph{finite integer index} (FII, defined below). It is used to argue that for every constant boundary size~$t$, there is a finite set of representatives~$\mathcal{R}_t$ such that any $t$-boundaried subgraph~$H$ can safely be replaced by some representative~$H' \in \mathcal{R}_t$, as described above. The key issue is that FII only guarantees that a finite set of representatives exist, without showing how to find it, how large the set is, or how many vertices the representative subgraphs have.

To deal with the issue of non-constructivity, Garnero et al.~\cite{GarneroPST15} introduced a framework based on dynamic programming. They showed that explicit bounds for the sizes of representatives can be obtained by analyzing the number of states required to solve the problem on graphs of bounded treewidth. By presenting explicit dynamic programming algorithms for problems such as $r$-\IndependentSet and $r$-\DominatingSet, they were able to derive upper bounds on the size of representatives in terms of the boundary size~$t$. These upper bounds grow very quickly with~$t$, in some cases triple-exponentially. Garnero et al.~\cite[\S 7]{GarneroPST15} suggest to examine to what extent this exponential dependance is unavoidable. We pursue this direction by presenting lower bounds.

\subparagraph{Boundaried graphs and equivalence}
To state our results we have to introduce some terminology.\footnote{To avoid an abundance of cumbersome definitions, our terminology differs slightly from that in earlier work (cf.~\cite{BodlaenderFLPST13,BodlaenderF01},~\cite[\S 2]{Fluiter97}). In particular, we do not allow $t$-boundaried graphs with fewer than~$t$ boundary vertices. The fact that we consider optimization problems as in~\cite{Fluiter97}, rather than decision problems as in~\cite{BodlaenderFLPST09,GarneroPST15}, forms no essential difference; our lower bounds also apply to those settings.} We only consider undirected, finite, simple graphs. Let~$t$ be a positive integer. A $t$-boundaried graph~$G$ consists of a vertex set~$V(G)$, an edge set~$E(G) \subseteq \binom{V(G)}{2}$, and an injective labeling~$\lambda_G \colon \{1, \ldots, t\} \to 
V(G)$ that identifies~$t$ distinct \emph{boundary vertices} in the graph. The \emph{boundary} of the graph is the set~$B_G := \{\lambda_G(1), \ldots, \lambda_G(t)\}$. Two $t$-boundaried graphs~$G$ and~$H$ can be \emph{glued} together on their boundary, resulting in the boundaried graph~$G \oplus H$ that is obtained from the disjoint union of~$G$ and~$H$ by identifying corresponding boundary vertices and removing any parallel edges that are introduced. That is, we merge~$\lambda_G(i)$ with~$\lambda_H(i)$ for each~$i \in [t]$. An \emph{optimization problem}~$\Pi$ on graphs assigns to every (unboundaried) graph~$G$ an optimal solution value~$\Pi(G) \in \mathbb{Z}$. We will also write~$\Pi(G)$ for a boundaried graph~$G$ to denote the optimum of the underlying unboundaried graph. Two $t$-boundaried graphs~$G$ and~$H$ are \emph{equivalent} with respect to~$\Pi$, denoted~$G \equiv_{\Pi,t} H$, if there exists a \emph{transposition constant}~$\Delta \in \mathbb{Z}$ such that for every $t$-boundaried graph~$F$:
\begin{equation} \label{eq:equivalence}
\Pi(G \oplus F) = \Pi(H \oplus F) + \Delta.$$
\end{equation}
It is easy to see that~$\equiv_{\Pi,t}$ is an equivalence relation. Problem~$\Pi$ has \emph{finite integer index} if~$\equiv_{\Pi,t}$ has a finite number of equivalence classes for each fixed~$t$. In the remainder, we omit the subscript~$t$ when it is clear from the context. Observe that these notions formalize the idea behind protrusion replacement sketched above: if $G \equiv_{\Pi,t} H$, then replacing~$G$ by~$H$ changes the optimum by exactly~$\Delta$.

\subparagraph{Our results}
We analyze the canonical equivalence relation~$\equiv_{\IS,t}$ on $t$-boundaried graphs for the \IndependentSet (\IS) problem, which asks for the maximum size of an independent set of pairwise non-adjacent vertices. We focus on \IndependentSet due to its simple combinatorial structure and develop most of the ideas in that context. Afterward, we will use a simple reduction to transfer these lower bounds to the \DominatingSet problem; the lower bounds described below also apply to \DominatingSet. 

Define a \emph{set of representatives} for~$\equiv_{\IS,t}$ to be a set~$\mathcal{R}_t$ of $t$-boundaried graphs, such that for every $t$-boundaried graph~$G$ there exists~$H \in \mathcal{R}_t$ with~$G \equiv_{\IS,t} H$. Let the \emph{critical size} of a set of representatives be the number of vertices of its largest graph. We aim to give a lower bound on the critical size of any set of representatives for \IndependentSet in terms of~$t$. Our approach consists of two steps. First, we construct a large set of pairwise nonequivalent graphs to give a lower bound on the number of equivalence classes of~$\equiv_{\IS,t}$. Then we use a counting argument to leverage this into a lower bound on the critical size. Observe that each equivalence class must be represented by a different graph. It follows that if the number of distinct $t$-boundaried graphs with at most~$s$ vertices is smaller than the number of equivalence classes, then the critical size of any set of representatives must be larger than~$s$ to give each class a distinct representative. By relating the number of small graphs to the number of equivalence classes, we therefore obtain the desired lower bounds.

Protrusion replacement is often applied in the context of restricted graph classes, where the protrusions to be replaced are known to have bounded treewidth and may even belong to a family of embeddable graphs such as planar graphs. With these application areas in mind, we develop our lower bounds to apply even when we wish only to have a representative for each equivalence class that contains a planar graph whose treewidth is~$t + \Oh(1)$, for boundary size~$t$. To find a large set of nonequivalent graphs we adapt a construction of Lokshtanov et al.~\cite{LokshtanovMS11a}, which they used to prove that \IndependentSet on graphs of treewidth~$w$ cannot be solved in time~$\Oh^*((2-\varepsilon)^w)$ for any~$\varepsilon > 0$ unless the Strong Exponential Time Hypothesis fails. We show that the graphs they construct can be made planar while increasing the treewidth (and in fact the pathwidth) by only a small additive term. More importantly, we show how to use this adapted construction to build a set of~$M(t) - 2$ planar graphs of small treewidth which are pairwise nonequivalent under~$\equiv_{\IS,t}$, for all~$t$. The term~$M(t) \geq 2^{\binom{t}{\lfloor t/2 \rfloor}} \geq 2^{2^t/\sqrt{4t}}$ denotes the $t$-th \emph{Dedekind number}, which counts the number of \emph{monotone Boolean functions} of~$t$ variables. The number of equivalence classes therefore grows double-exponentially with~$t$. Using the counting argument above, this allows us to give a lower bound of~$\Omega(\log M(t)) \geq \Omega(2^t / \sqrt{4t})$ on the critical size of any set of planar representatives for the equivalence classes of~$\equiv_{\IS,t}$ that contain a planar graph of bounded pathwidth.

While developing a lower bound on the number of equivalence classes for planar graphs of bounded pathwidth, we also found an exact characterization of the number of equivalence classes of~$\equiv_{\IS,t}$ in general. We define a natural class of functions from~$\{0,1\}^t$ to~$\mathbb{N}$ that we call \emph{$t$-representative functions}. We give a bijection between the $t$-representative functions and the equivalence classes of~$\equiv_{\IS,t}$ for $t$-boundaried graphs. As we will show that all monotone Boolean functions which are not constantly zero yield a distinct $t$-representative function, this gives a lower bound of~$M(t) - 1$ on the number of equivalence classes of~$\equiv_{\IS,t}$. On the other hand, we show that the number of such functions is at most~$2^{2^t-1}$. The double-exponential lower bound for the number of equivalence classes containing a bounded-pathwidth planar graph is therefore not far off from the upper bound of~$2^{2^t-1}$ in general graphs. The fact that the base of the double-exponential in this expression is independent of~$t$ is noteworthy. The naive way to bound the number of equivalence classes is to associate a table to each $t$-boundaried graph. For each subset~$S$ of the boundary vertices~$B$, the table stores the maximum size of an independent set containing no vertex of~$B \setminus S$. There are at most~$t+1$ distinct values in such a table, and two boundaried graphs whose tables differ in the same universal constant in all positions are easily shown to be equivalent. As there are~$2^t$ entries in the table, and~$t+1$ different options per entry, this gives an upper bound of~$(t+1)^{2^t}$ on the number of equivalence classes. Garnero et al.~\cite[Lemma 3.7]{GarneroPST15} obtain the same bound using a subtly different definition for the table. Our result of~$2^{2^t-1}$ yields a slight improvement.

\subparagraph{Organization}
In Section~\ref{sec:preliminaries} we present preliminaries on graphs and Boolean functions. Section~\ref{sec:characterizingclasses} presents a simple characterization of the equivalences of~$\equiv_{\IS}$ in terms of \emph{$t$-representative Boolean functions}, thereby providing an upper bound on the number of equivalence classes. In Section~\ref{sec:defining:general} we show that for each $t$-representative Boolean function~$f$, one can construct a $t$-boundaried graph whose equivalence class under~$\equiv_{\IS}$ encodes~$f$, thereby establishing a bijection between equivalence classes and $t$-representative functions. In Section~\ref{sec:counting} we give upper and lower bounds on the number of $t$-representative functions, and thereby on the number of equivalence classes of~$\equiv_{\IS}$. We construct a large set of $\equiv_{\IS}$-nonequivalent \emph{planar} graphs in Section~\ref{sec:treewidthgraphs}. This construction is combined with a counting argument in Section~\ref{sec:lowerbounds:protrusions} to give a lower bound on the critical size of representatives for \IndependentSet. Using a simple reduction, these lower bounds are transferred to \DominatingSet in Section~\ref{sec:domset}.

\section{Preliminaries}\label{sec:definitions} \label{sec:preliminaries}
We use~$\mathbb{N}$ to denote the natural numbers, including~$0$. For a positive integer~$n$ and a set~$X$ we use~$\binom{X}{n}$ to denote the collection of all subsets of~$X$ of size~$n$. The \emph{power set} of~$X$ is denoted~$2^{X}$. The set~$\{1, \ldots, n\}$ is abbreviated as~$[n]$. A \emph{Boolean function} is a function of the form~$f \colon \{0,1\}^n \to \{0,1\}$. We sometimes use the equivalent view that a Boolean function assigns a $0/1$-value to every subset~$S \subseteq [n]$, which is the value of~$f$ when the arguments whose index is in~$S$ are set to~$1$ and the remaining arguments are set to~$0$. A Boolean function~$f \colon 2^{[n]} \to \{0,1\}$ is \emph{monotone} if~$f(S') \leq f(S)$ whenever~$S' \subseteq S \subseteq [n]$. We will call Boolean functions in this form set-functions, and may replace~$[n]$ by other finite sets of ordered elements. A formula in conjunctive normal form (CNF) is \emph{monotone} if no literal appears negated. 

\begin{proposition} \label{prop:monotone:cnf}
For every non-constant monotone Boolean set-function~$f \colon 2^{[n]} \to \{0,1\}$ there is a monotone CNF formula~$\phi$ such that for all~$x_1, \ldots, x_n \in \{0,1\}^n$ we have~$\phi(x_1, \ldots, x_n) = 1$ if and only if~$f(\{i \mid x_i = 1\}) = 1$.
\end{proposition}
\begin{proof}
Consider a monotone set-function~$f$. Let~$\mathcal{S} \subseteq 2^{[n]}$ denote the inclusion-wise maximal subsets~$S$ of~$[n]$ for which~$f(S) = 0$. Since~$f$ is not constantly~$1$, the set~$\mathcal{S}$ is not empty. Since~$f$ is monotone, it follows that~$f(S) = 0$ if and only if~$S \subseteq T$ for some~$T \in \mathcal{S}$. Create a CNF~$\phi := \bigwedge _{T \in \mathcal{S}} \bigvee _{i \in [n] \setminus T} x_i$. Since~$f$ is not constantly~$0$, we have~$[n] \not \in \mathcal{S}$ and each clause in~$\phi$ has at least one literal. For every~$S \subseteq [n]$ with~$f(S) = 1$ we know that~$S \not \subseteq T$ for all~$T \in \mathcal{S}$ and therefore setting~$x_i$ to~$1$ for all~$i \in S$ and the remaining variables to~$0$ satisfies~$\phi$. Conversely, suppose that setting all variables~$x_i$ with~$i \in S$ to~$1$ and the remainder to~$0$ satisfies the CNF. Then for every~$T \in \mathcal{S}$ we have~$S \not \subseteq T$ and therefore~$f(S) = 1$. Since~$\phi$ only has positive literals, this concludes the proof.
\end{proof}

\subparagraph{Graphs}
We will denote the treewidth of a graph~$G$ by $\tw(G)$ and its pathwidth by~$\pw(G)$. It is well-known that~$\pw(G) \geq \tw(G)$; refer to a textbook for further details~\cite[\S 7]{CyganFKLMPPS15}. We use~$\opt_{\IS}(G)$ to denote the size of a maximum independent set in a graph~$G$. By~$\opt_\VC(G)$ we denote the size of a minimum vertex subset that intersects all edges (a vertex cover). Finally,~$\opt_\DS(G)$ denotes the minimum size of a vertex subset such that any vertex not in the subset, has a neighbor in the subset (a dominating set). We use the following consequence of the gluing operation.

\begin{proposition} \label{prop:gluesets}
Let~$G$ and~$H$ be $t$-boundaried graphs that share the same set of boundary vertices~$B = \{v_1, \ldots, v_t\}$ but are otherwise vertex-disjoint. Then a vertex set~$X \subseteq V(G \oplus H)$ is independent in~$G \oplus H$ if and only if~$X \cap V(G)$ is independent in~$G$ and~$X \cap V(H)$ is independent in~$H$.
\end{proposition}
\begin{proof}
\textbf{$(\Rightarrow$)} Since~$G$ and~$H$ are both subgraphs of~$G \oplus H$, the vertices of~$X$ that belong to a given subgraph form an independent set in that subgraph.

\textbf{$(\Leftarrow$)} The definition of gluing implies that every edge of~$G \oplus H$ is an edge of~$G$, an edge of~$H$, or both. If~$X$ is not independent in~$G \oplus H$ because it induces an edge of~$G$, then~$X \cap V(G)$ is not independent in~$G$. If~$X$ induces an edge of~$H$, then~$X \cap V(H)$ is not independent in~$H$. As any edge in~$G \oplus H$ belongs to one of these two categories, this concludes the proof.
\end{proof}

\section{Characterizing equivalence classes for Independent Set} \label{sec:characterizingclasses}
In this section we derive several tools to analyze the equivalence classes of~$\equiv_{\IS}$. For each $t$-boundaried graph~$G$ we define a function that captures the interaction of optimal independent sets with its boundary. These will be useful to reason about the (non)equivalence of pairs of graphs with respect to~$\equiv_{\IS}$.

\begin{definition} \label{def:subsetfunc}
Let~$G$ be a $t$-boundaried graph with boundary~$B = \{v_1, \ldots, v_t\}$. 
The function~$\s_G \colon 2^B \to \mathbb{N}$ expresses the size of a maximum independent set in~$G$ whose intersection with the boundary is a \emph{subset} of a given set:
\begin{equation}
\s_G(S) := \max \bigl \{ |X| \bigmid X \text{ is an independent set in~$G$ with } X \cap B \subseteq S \bigr \}.
\end{equation}
\end{definition}

We will see that equivalence classes can be characterized by the functions~$\s_G$ of the graphs~$G$ in that class. The next lemma shows that when gluing two boundaried graphs~$G$ and~$H$ together, the optimum of the resulting graph~$G \oplus H$ can be deduced from~$\s_G$ and~$\s_H$. The identity we prove is reminiscent of the recurrence that is used for \textsc{join} nodes when solving \IndependentSet on graphs of bounded treewidth~\cite[\S 7.3.1]{CyganFKLMPPS15}.

\begin{lemma} \label{lemma:opt:from:functions}
Let~$G$ and~$H$ be $t$-boundaried graphs for some~$t$. The following holds:
$$\max _{S\subseteq B}\{\s_G(S) + \s_H(S) - |S|\} = \optis(G \oplus H).$$
\end{lemma}
\begin{proof}
Assume without loss of generality that~$G$ and~$H$ have the same set of vertices~$B = \{v_1, \ldots, v_t\}$ as their boundary and are otherwise vertex-disjoint:~$V(G) \cap V(H) = B$. We prove equality by showing that the inequality holds in both directions.

\textbf{$(\leq)$} Consider a set~$S^* \subseteq B$ maximizing the expression on the left. Let~$X_G \subseteq V(G)$ and~$X_H \subseteq V(H)$ be independent sets in~$G$ and~$H$ of sizes~$\s_G(S^*)$ and~$\s_H(S^*)$ with~$X_G \cap B \subseteq S^*$ and~$X_H \cap B \subseteq S^*$; these exist by Definition~\ref{def:subsetfunc}. Consider the multiset~$X'$ obtained by taking the disjoint union of~$X_G$ and~$X_H$, which contains elements of~$S^*$ twice if they are used in both~$X_H$ and~$X_G$. Note that~$X'$ contains exactly~$\s_G(S^*) + \s_H(S^*)$ occurrences of vertices (possibly duplicating elements of~$S^*$). Let~$X$ be the result of removing from~$X'$ one copy of each vertex in~$S^*$ (if such a copy is present at all). Then~$X$ is a simple set (no vertex occurs more than once since we removed one copy of each vertex that could occur twice) and~$|X| \geq \s_G(S^*) + \s_H(S^*) - |S^*|$. To show that~$\optis(G \oplus H)$ is at least the value of the left-hand side, it therefore suffices to show that~$\optis(G \oplus H) \geq |X|$. To conclude the proof we therefore argue that~$X$ is an independent set in~$G \oplus H$. Since vertices of~$S^*$ that occur in only one of the sets~$X_G,X_H$ do not occur in~$X$ (their single occurrence is removed from~$X'$ in the process of constructing~$X$), we have~$X \cap V(G) \subseteq X_G$ and~$X \cap V(H) \subseteq X_H$. Since~$X_G$ and~$X_H$ are independent in~$G$ and~$H$, Proposition~\ref{prop:gluesets} implies that~$X$ is independent in~$G \oplus H$, which concludes this direction of the proof.

\textbf{$(\geq)$} Let~$X$ be a maximum independent set in~$G \oplus H$ of size~$\optis(G \oplus H)$, and let~$S^* := X \cap B$ be the vertices used in the boundary. Then~$\s_G(S^*) \geq |X \cap V(G)|$ and~$\s_H(S^*) \geq |X \cap V(H)|$, by definition. Observe that~$|X| = |X \cap V(G)| + |X \cap V(H)| - |S^*|$, since~$X \cap V(G) \cap V(H)  = S^*$ and so those vertices are counted twice. So~$\s_G(S^*) + \s_H(S^*) - |S^*| \geq |X \cap V(G)| + |X \cap V(H)| - |S^*| = |X| = \optis(G)$. As the maximum of the left-hand side over all sets~$S \subseteq B$ is at least as large as the value for~$S^*$, this concludes the proof.
\end{proof}

To relate the equivalence of graphs to properties of the corresponding functions~$\s$, the following indicator graphs will be convenient.

\begin{definition} \label{def:indicator}
Let~$t$ be a positive integer and~$B = \{v_1, \ldots, v_t\}$. For each $S\subseteq B$ define the $t$-boundaried \emph{indicator} graph~$I_S$ with boundary~$B$ as the result of the following process: starting from an edgeless graph with vertex set $B$, for each $v_i \in B \setminus S$ add vertices $u_i, u'_i$ and the edges $\{v_i, u_i\}, \{v_i, u'_i\}$ to $I_S$.
\end{definition}

Each boundary vertex not in~$S$ thus becomes the center of a star with two leaves in~$I_S$, and boundary vertices in~$S$ are isolated vertices in~$I_S$. The next proposition shows that maximum independent sets of~$F \oplus I_S$ reveal the value of~$\s_F(S)$.

\begin{proposition} \label{prop:glueindicator}
$\optis(F \oplus I_S) = \s_F(S) + 2(t - |S|)$ for all $t$-boundaried graphs~$F$.
\end{proposition}
\begin{proof}
Let~$U := \{u_i, u'_i \mid v_i \in B \setminus S\}$ denote the $u$-vertices of the graph~$I_S$; note that~$|U| = 2(t - |S|)$. We prove the equality by establishing two inequalities.

\textbf{($\geq$)} Let~$X_F$ be an independent set in~$F$ of size~$\s_F(S)$ with~$X_F \cap B \subseteq S$, which exists by definition of~$\s_F$. Then~$U \cup X_F$ is an independent set in~$F \oplus I_S$ of size~$\s_F(S) + |U| = \s_F(S) + 2(t - |S|)$, so the optimal independent set is at least as large.

\textbf{($\leq$)} Let~$X$ be a maximum independent set in~$F \oplus I_S$. We claim that~$U \subseteq X$. To see that, observe that the vertices~$u_i$ and~$u'_i$ have degree one in~$F \oplus I_S$, so if they are not in a maximal independent set then their common neighbor~$v_i$ is; but then one can obtain a larger independent set by replacing~$v_i$ by~$u_i$ and~$u'_i$. It follows that~$U \subseteq X$, and by the edges between~$u_i$ and~$v_i$ it follows that~$v_i \not \in X$ for all~$v_i \in B \setminus S$. In other words:~$X$ contains no vertex of~$B \setminus S$. Since~$V(F \oplus I_S) = U \cup V(F)$, this implies that~$X_F := X \setminus U$ is an independent set in~$F$ of size~$|X_F| = |X| - |U| = \optis(F \oplus I_S) - 2(t-|S|)$. It contains no vertex of~$B \setminus S$, implying that its intersection with the boundary is a subset of~$S$. It follows that~$\s_F(S) \geq |X| - |U| = \optis(F \oplus I_S) - 2(t - |S|)$. Adding~$2(t-|S|)$ on both sides and reversing the direction yields the desired inequality.
\end{proof}

Using Proposition~\ref{prop:glueindicator} we can show that the equivalence class of a boundaried graph~$G$ with respect to~$\equiv_{\IS}$ is completely characterized by the function~$\s_G$.

\begin{theorem} \label{thm:fiffequivalent}
Let $G$ and~$H$ be two $t$-boundaried graphs with boundary $B = \{v_1, \ldots, v_t\}$. Then~$G \equiv_{\IS,t} H$ if and only if there exists a constant $c \in \mathbb{Z}$ such that $\s_{G}(S) = \s_{H}(S) + c$ for all $S\subseteq B$.
\end{theorem}
\begin{proof}
We prove the two directions of the equivalence separately.

\textbf{($\Leftarrow$)} Assume that~$\s_G(S) = \s_H(S) + c$ for all~$S \subseteq B$. To prove that~$G \equiv_\IS H$ as per Equation~\ref{eq:equivalence}, it suffices to show that for all $t$-boundaried graphs~$F$ we have~$\optis(G \oplus F) = \optis(H \oplus F) + c$. Now observe:
\begin{align*}
\optis(G \oplus F) &= \max_{S\subseteq B}\{\s_G(S) + \s_F(S) - |S|\} & \mbox{By Lemma~\ref{lemma:opt:from:functions}.} \\
&= \max_{S\subseteq B}\{\s_H(S) + c + \s_F(S) - |S|\} & \mbox{By assumption.} \\
&= \max_{S\subseteq B}\{\s_H(S) + \s_F(S) - |S|\} + c\\
&= \optis(H \oplus F) + c & \mbox{By Lemma~\ref{lemma:opt:from:functions}.}
\end{align*}

\textbf{($\Rightarrow$)} Suppose that~$G \equiv_\IS H$ and let~$c$ be a constant such that~$\optis(G \oplus F) = \optis(H \oplus F) + c$ for all~$F$. Now consider an arbitrary~$S \subseteq B$:
\begin{align*}
\s_{G}(S) &= \optis(G\oplus I_S) - 2(t - |S|) & \mbox{By Proposition~\ref{prop:glueindicator}.} \\
&= \optis(H\oplus I_S) + c - 2(t - |S|) & \mbox{Since~$G \equiv_\IS H$.} \\
&= \optis(H\oplus I_S) - 2(t - |S|) + c \\
&= \s_{H}(S) + c & \mbox{By Proposition~\ref{prop:glueindicator}.}
\end{align*}
This concludes the proof of Theorem~\ref{thm:fiffequivalent}.
\end{proof}

Theorem~\ref{thm:fiffequivalent} shows that two $t$-boundaried graphs~$G$ and~$H$ are equivalent under~$\equiv_\IS$ if the functions~$\s_G$ and~$\s_H$ differ by a fixed constant for all inputs. It will be convenient to eliminate this degree of freedom by normalizing the functions.

\begin{definition} \label{def:normalized}
The \emph{normalized boundary function} of a $t$-boundaried graph~$G$ with boundary~$B$ is the function~$\s^0_G \colon 2^B \to \mathbb{N}$ given by~$\s^0_G(S) := \s_G(S) - \s_G(\emptyset)$.
\end{definition}

Intuitively,~$\s^0_G(S)$ represents how much larger an independent set can be if we are allowed to use the boundary vertices from~$S$, compared to when we are not allowed to use any boundary vertices in the independent set.

\begin{corollary} \label{cor:characteristics}
Let~$G$ and~$H$ be two $t$-boundaried graphs with boundary~$B = \{v_1, \ldots, v_t\}$. Then~$G \equiv_\IS H$ if and only if~$\s^0_G = \s^0_H$.
\end{corollary}
\begin{proof}
By Theorem~\ref{thm:fiffequivalent} it suffices to prove that~$\s^0_G = \s^0_H$ if and only if there is a constant~$c$ such that $\s_{G}(S) = \s_{H}(S) + c$ for all $S\subseteq B$. For the forward direction, it is easy to verify that choosing~$c := \s_G(\emptyset) - \s_H(\emptyset)$ suffices. For the reverse direction, suppose that~$\s_G(S) = \s_H(S) + c$ for all~$S \subseteq B$. Then for all~$S \subseteq B$ we have:
$$\s^0_G(S) = \s_G(S) - \s_G(\emptyset) = (\s_H(S) + c) - (\s_H(\emptyset) + c) = \s_H(S) - \s_H(\emptyset) = \s^0_H(S).\qedhere$$
\end{proof}

Corollary~\ref{cor:characteristics} shows that equivalence classes of~$\equiv_\IS$ are determined by the normalized boundary functions of the graphs in the class. To see how many different equivalence classes there can be, it is therefore useful to analyze the properties of normalized boundary functions.

\begin{definition} \label{def:representativefunctions}
Let~$t$ be a positive integer and let~$B := \{v_1, \ldots, v_t\}$. A function~$f \colon 2^B \to \mathbb{N}$ is called a \emph{$t$-representative function} if it satisfies the following three properties:
\begin{enumerate}
	\item $f(\emptyset) = 0$.\label{pty:startzero}
	\item \emph{Monotonicity}: for any~$S' \subseteq S \subseteq B$ we have~$f(S') \leq f(S)$.\label{pty:monotone}
	\item \emph{Bounded increase}: For every nonempty set~$S \subseteq B$ we have~$f(S) \leq 1 + \min _{v \in S} f(S \setminus \{v\})$.\label{pty:boundedincrease}
\end{enumerate}
\end{definition}

\begin{lemma} \label{lemma:boundary:char:is:representative}
Let~$G$ be a $t$-boundaried graph with boundary~$B := \{v_1, \ldots, v_t\}$. Then~$\s^0_G$ is a $t$-representative function.
\end{lemma}
\begin{proof}
We prove that~$\s^0_G$ has the three properties given in Definition~\ref{def:representativefunctions}.

\textbf{(\ref{pty:startzero})} By definition of~$\s^0_G$ we have~$\s^0_G(\emptyset) = \s_G(\emptyset) - \s_G(\emptyset) = 0$.

\textbf{(\ref{pty:monotone})} This follows directly from Definitions~\ref{def:subsetfunc} and~\ref{def:normalized}: the collection of independent sets over which~$\s_G(S')$ optimizes is a subset of the independent sets over which~$\s_G(S)$ optimizes.

\textbf{(\ref{pty:boundedincrease})} Consider a nonempty set~$S \subseteq B$ and let~$X$ be an independent set in~$G$ of size~$\s_G(S)$ with~$X \cap B \subseteq S$, which exists by Definition~\ref{def:subsetfunc}. For every~$v \in S$ we have that~$X \setminus \{v\}$ is an independent set of size~$|X| - 1$ in~$G$ whose intersection with~$B$ is a subset of~$S \setminus \{v\}$, implying that~$\s_G(S \setminus \{v\}) \geq |X| - 1 = \s_G(S) - 1$. Adding~$1 - \s_G(\emptyset)$ on both sides we obtain~$\s^0_G(S) = \s_G(S) - \s_G(\emptyset) \leq 1 + \s_G(S \setminus \{v\}) - \s_G(\emptyset) = 1 + \s^0_G(S \setminus \{v\})$. As this holds for all~$v \in S$, it holds in particular for~$v \in S$ minimizing~$\s^0_G(S \setminus \{v\})$.
\end{proof}

\section{Defining graphs with given boundary characteristics} \label{sec:defining:general}

Corollary~\ref{cor:characteristics} shows that $t$-boundaried graphs with the same normalized boundary function belong to the same equivalence class. Since each normalized boundary function is a $t$-representative function by Lemma~\ref{lemma:boundary:char:is:representative}, this implies that the number of equivalence classes of~$\equiv_{\IS,t}$ is \emph{at most} the number of distinct $t$-representative functions. In Lemma~\ref{lemma:generalgraphs} we will show that, surprisingly, the converse also holds: for each $t$-representative function there is a distinct equivalence class of~$\equiv_{\IS,t}$. Before proving that lemma, we first derive a useful property of $t$-representative functions.

\begin{proposition} \label{prop:swap:subsets}
Each $t$-representative function~$f$ satisfies~$f(S') - |S' \setminus S| \leq f(S)$ for all~$S, S' \subseteq B$.
\end{proposition}
\begin{proof}
By Property~\ref{pty:boundedincrease}, every time we remove an element of~$S'$ the function value drops by at most one. If we remove the elements~$S' \setminus S$ one at a time from~$S'$ until arriving at the set~$S' \cap S$, we therefore decrease the value by at most~$|S' \setminus S|$. This implies that~$f(S \cap S') \geq f(S') - |S' \setminus S|$. Hence~$f(S') - |S' \setminus S| \leq f(S \cap S') \leq f(S)$, where the last step uses Property~\ref{pty:monotone}.
\end{proof}

\begin{lemma}\label{lemma:generalgraphs}
For every $t$-representative function~$f$, there exists a $t$-boundaried graph $G$ with boundary $B := \{v_1,v_2,\ldots,v_t\}$, such that $\s^0_G(S) = f(S)$ for every $S\subseteq B$.
\end{lemma}
\begin{proof}
Consider an arbitrary $t$-representative function~$f$, which assigns a non-negative integer to each~$S \subseteq B$. We construct a $t$-boundaried graph~$G$ for which~$\s^0_G = f$, as follows:

\begin{enumerate}
	\item Start from an edgeless graph with vertex set~$B$, which is the boundary of the graph.
	\item For each~$i \in [t]$ add a vertex~$u_i$ and the edge~$\{u_i, v_i\}$.
	\item For each~$S \subseteq B$ with~$f(S) > 0$, add a set~$V_S = \{v_{S,1}, \ldots, v_{S,f(S)}\}$ consisting of~$f(S)$ vertices to the graph. These vertices are false twins (all share the same open neighborhood) and are connected to the rest of the graph as follows:
	\begin{enumerate}
		\item For each~$i \in [t]$ with~$v_i \in S$, all vertices of~$V_S$ are adjacent to~$u_i$.
		\item For each~$i \in [t]$ with~$v_i \not \in S$, all vertices of~$V_S$ are adjacent to~$v_i$.
		\item All vertices of~$V_S$ are adjacent to all vertices~$V_{S'}$ that are created for sets~$S' \neq S$.
	\end{enumerate}
\end{enumerate}		

We show that~$\s_G(S) = t + f(S)$ for all~$S \subseteq B$. This will imply that~$\s^0_G(S) = \s_G(S) - \s_G(\emptyset) = (t + f(S)) - (t + f(\emptyset)) = (t + f(S)) - (t + 0) = f(S)$ for all~$S \subseteq B$, since~$f(\emptyset) = 0$ by Definition~\ref{def:representativefunctions}. We therefore conclude the proof by showing that~$\s_G(S) = t + f(S)$ for all~$S \subseteq B$, by establishing two inequalities. Consider an arbitrary~$S \subseteq B$.

\textbf{($\geq$)} To show~$\s_G(S) \geq t + f(S)$ we construct an independent set~$X$ in~$G$ of size~$t + f(S)$ that intersects~$B$ in a subset of~$S$. If~$f(S) = 0$ then~$X = \{u_1, \ldots, u_t\}$ suffices, so assume in the remainder that~$f(S) > 0$. Let~$X$ consist of the~$f(S)$ vertices in~$V_S$, together with the vertices~$\{u_i \mid i \in [t], v_i \not \in S\}$ and~$\{v_i \mid i \in [t], v_i \in S\}$. Then~$|X| = t + f(S)$, and using the construction above it is straight-forward to verify that~$X$ is an independent set. Since~$X \cap B = S$, this shows that~$\s_G(S) \geq t + f(S)$.

\textbf{($\leq$)} Now we argue that~$\s_G(S) \leq t + f(S)$. Consider a maximum independent set~$X$ in~$G$ that intersects~$B$ in a subset of~$S$, which has size~$\s_G(S)$ by Definition~\ref{def:subsetfunc}. If~$X$ contains no vertices of~$V_{S'}$ for any~$S' \subseteq B$, then~$X$ has at most~$t$ vertices: an independent set contains at most one vertex of each edge~$\{v_i, u_i\}$ for each~$i \in [t]$. Hence~$|X| \leq t$ in this case, which is at most~$t + f(S)$ since~$f(S) \geq 0$ by Properties~\ref{pty:startzero} and~\ref{pty:monotone}. In the remainder, assume~$X$ contains a vertex of~$V_{S'}$ for some~$S' \subseteq B$. This implies that~$X$ contains no vertices from~$V_{S''}$ for any~$S'' \neq S'$, since all vertices of~$V_{S'}$ are adjacent to all vertices of~$V_{S''}$ by construction of~$G$. Hence besides the vertices from~$V_{S'}$, the set~$X$ only contains vertices of edges~$\{v_i, u_i\}$ for~$i \in [t]$. The independent set~$X$ contains at most one vertex from each such edge. For each~$v_i \in S' \setminus S$, observe that~$X$ does not contain~$v_i$ (since~$X \cap B \subseteq S$), and~$X$ does not contain~$u_i$ either (since~$u_i$ is adjacent to all members of~$V_{S'}$). So~$X$ has at most~$f(S')$ vertices from~$V_{S'}$, no vertices of~$\{v_i, u_i\}$ for each~$v_i \in S' \setminus S$, and at most one vertex from each of the remaining~$t - |S' \setminus S|$ edges. It follows that~$|X| \leq f(S') + (t - |S' \setminus S|)$. By Proposition~\ref{prop:swap:subsets} we have~$f(S') - |S' \setminus S| \leq f(S)$, which shows that~$|X| \leq t + f(S)$ and concludes the proof.
\end{proof}

\section{Counting \texorpdfstring{$t$-representative functions}{t-representative functions}} \label{sec:counting}
We say that two $t$-representative functions are distinct if their function values differ on some input. Lemma~\ref{lemma:generalgraphs} shows that for each $t$-representative function~$f$, there exists a $t$-boundaried graph whose normalized boundary function equals~$f$. Together with Corollary~\ref{cor:characteristics}, which says that boundaried graphs with the same normalized boundary function are equivalent under~$\equiv_{\IS,t}$, this establishes a bijection between the equivalence classes of~$\equiv_{\IS,t}$ and the $t$-representative functions. To bound the number of equivalence classes of~$\equiv_{\IS,t}$  it therefore suffices to bound the number of $t$-representative functions. Recall that~$M(t)$ denotes the $t$-th \emph{Dedekind number}, the number of distinct monotone Boolean functions of~$t$ variables.

\begin{lemma} \label{lem:flowerbound}
There are at least~$M(t) - 1$ distinct $t$-representative functions.
\end{lemma}
\begin{proof}
Consider a monotone Boolean function~$g \colon \{0,1\}^t \to \{0,1\}$, and the derived set-function~$f \colon 2^{[t]} \to \{0,1\}$ as described in Section~\ref{sec:preliminaries}. If~$f$ is not constantly~$1$ (causing it to violate Property~\ref{pty:startzero}), then it is a $t$-representative function since it is monotone by definition, while Property~\ref{pty:boundedincrease} is trivial when the range is~$\{0,1\}$. Hence all the~$M(t) - 1$ monotone Boolean functions that are not constantly~$1$ yield a distinct $t$-representative function.
\end{proof}

It is known that~$M(t) \geq 2^{\binom{t}{\lfloor t/2 \rfloor}}$. To see this, consider the subsets~$\mathcal{S}_t = \binom{[t]}{\lfloor t/2 \rfloor}$ of~$[t]$ of size~$\lfloor t / 2 \rfloor$. For each subset~$\mathcal{S}'_t \subseteq \mathcal{S}_t$ we obtain a different monotone set-function by saying that~$f(S) = 1$ if and only if~$S$ contains one of the subsets in~$\mathcal{S}'_t$. By Stirling's approximation we have~$\binom{t}{\lfloor t/2 \rfloor} \geq 2^t / \sqrt{4 t}$, which implies that~$M(t) \geq 2^{2^t / \sqrt{4t}}$. The following lemma gives an upper bound on the number of $t$-representative functions.

\begin{lemma} \label{lem:fupperbound}
The number of distinct $t$-representative functions is at most~$2^{2^t - 1}$.
\end{lemma}
\begin{proof}
For a $t$-representative function~$f$, consider the set-function~$f' \colon 2^{[t]} \to \{0,1\}$ given by:
$$f'(S) = \begin{cases}
0 & \mbox{if~$S = \emptyset$,} \\
f(S) - \min_{v \in S}f(S \setminus \{v\}) & \mbox{otherwise.}
\end{cases}$$
By Properties~\ref{pty:monotone} and~\ref{pty:boundedincrease} the function~$f'$ indeed takes values in~$\{0,1\}$. We show that we can recover~$f$ from~$f'$. Define the function~$f^*$ recursively, as follows:
$$f^*(S) = \begin{cases} 0 & \mbox{if~$S = \emptyset$,} \\
f'(S) + \min_{v \in S}f^*(S \setminus \{v\}) & \mbox{otherwise.} \\
\end{cases}$$
It is easy to verify that~$f^* = f$ and therefore that~$f'$ completely characterizes~$f$. The number of $t$-representative functions is therefore bounded by the number of distinct derived functions~$f'$. Since there are~$2^{2^t-1}$ distinct ways to choose the~$0/1$-value of~$f'$ on the~$2^t - 1$ nonempty subsets of~$[t]$ (note that~$f'(S) = 0$ is fixed), this gives the desired upper bound.
\end{proof}

Lemmata~\ref{lem:flowerbound} and~\ref{lem:fupperbound} give the following corollary for each positive integer~$t$.

\begin{corollary}
The number of equivalence classes of~$\equiv_{\IS,t}$ lies between~$2^{2^t/\sqrt{4t}}$ and~$2^{2^t-1}$.
\end{corollary}

\section{Defining planar graphs with given boundary characteristics}\label{sec:treewidthgraphs}

In Lemma~\ref{lemma:generalgraphs} we constructed nonequivalent $t$-boundaried graphs based on distinct $t$-representative functions. The graphs constructed in that lemma have large treewidth and are far from being planar; they contain cliques of size roughly~$2^t$. To derive lower bounds that are meaningful even when protrusion replacement is applied for planar graphs of bounded treewidth, we present an alternative construction to lower bound the number of equivalence classes that contain a planar graph of small pathwidth (and therefore have small treewidth). The following gadget, of which several variations were used in earlier work (cf.~\cite[Theorem 5.3]{Jansen13b} and~\cite{FominS16,LokshtanovMS11a}), will be useful in our construction.

\begin{definition}
Let~$k$ be a positive integer. The \emph{clause gadget of size~$k$} is the graph~$\mathcal{C}_k$ constructed as follows (see Figure~\ref{fig:clausegadget}). For each~$i \in [k]$ create a triangle on vertices~$\{u_i, v_i,w_i\}$. Connect these into a path by adding all edges~$\{w_i, u_{i+1}\}$ for~$i \in [k-1]$. Finally, add vertices~$\vstart,\vend$ and the edges~$\{\vstart, u_1\}$ and~$\{w_k, \vend\}$. The vertices~$(v_1, \ldots, v_k)$ are the \emph{terminals} of the clause gadget.
\end{definition}

\begin{figure}[t]
	\centering
	\begin{subfigure}{0.30\textwidth}
		\includegraphics[height=2.5cm]{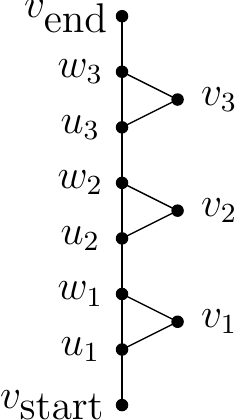}
		\caption{Clause gadget $\mathcal{C}_3$}
		\label{fig:clausegadget}
	\end{subfigure}
	\begin{subfigure}{0.30\textwidth}
		\includegraphics[height=2.5cm]{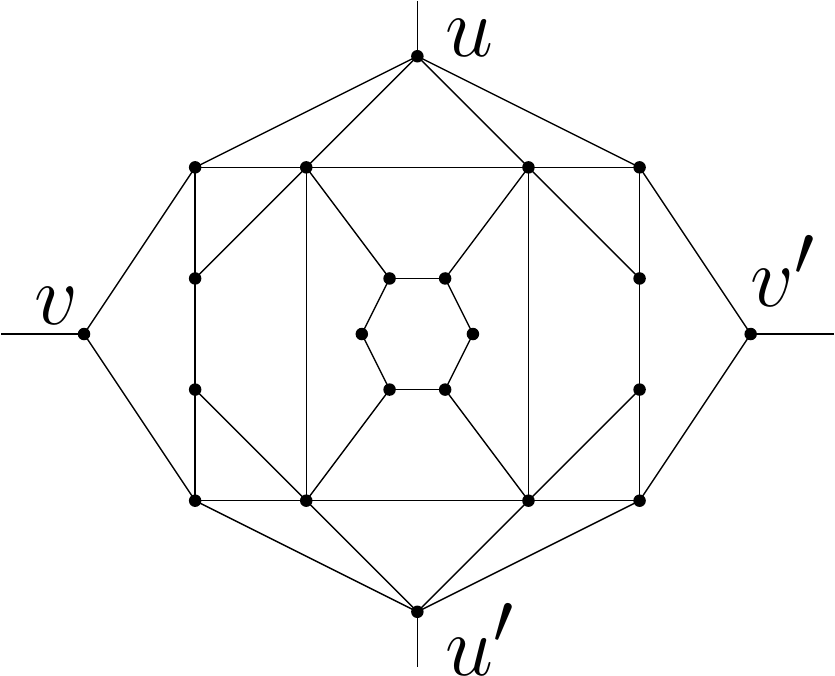}
		\caption{Crossover gadget $G_{\times}$}
		\label{fig:crossover}
	\end{subfigure}
	\begin{subfigure}{0.25\textwidth}
		\centering
		\begin{tabular}{|c|c|c|c|}\hline
			\backslashbox{$j$}{$i$} & $0$ & $1$ & $2$\\\hline
			$0$ & $7$ & $8$ & $8$ \\\hline
			$1$ & $8$ & $9$ & $9$ \\\hline
			$2$ & $7$ & $8$ & $9$ \\\hline
		\end{tabular}
		\caption{Values of $\optis$ in~$G_\times$}
		\label{tab:crossover}
	\end{subfigure}
	\caption{Gadgets for \IndependentSet. The crossover gadget is due to Garey et al.~\cite[Fig.~11 and Table 1]{GareyJS76}. The table on the right shows for all relevant combinations of~$i$ and~$j$ what the maximum size is of an independent set~$X$ satisfying~$|\{v,v'\}\cap X| = i$ and~$|\{u,u'\}\cap X| = j$.}
	\label{fig:crossovergadget}
\end{figure}

\begin{observation}[{Cf. \cite[Obs. 6--8]{FominS16}}] \label{obs:clausegadget}
For each positive~$k \in \mathbb{N}$, the clause gadget~$\mathcal{C}_k$ has the following properties:
\begin{enumerate}
	\item $\optis(\mathcal{C}_k) = k + 2$.
	\item Every maximum independent set in~$\mathcal{C}_k$ contains a terminal vertex~$v_i$ for some~$i \in [k]$.\label{clausegadget:chooses:terminal}
	\item $\forall i \in [k]$ there is a maximum independent set in~$\mathcal{C}_k$ containing~$v_i$ but no other terminals.\label{clausegadget:oneterminal}
	\item $\mathcal{C}_k$ is planar and~$\pw(\mathcal{C}_k) = 2$.
\end{enumerate}
\end{observation}

To ensure our construction yields a planar graph, we use a \emph{crossover gadget} for \IndependentSet due to Garey et al.~\cite{GareyJS76}. It was originally designed for \VertexCover, but since the complement of a maximum independent set is a minimum vertex cover, we can rephrase the properties of the gadget in terms of independent sets. The crossover gadget $G_{\times}$ is the 22-vertex graph illustrated in Figure~\ref{fig:crossover}, which has four terminals~$(u,u',v,v')$. When we have a drawing of a graph~$G$ in which exactly two edges $\{a,b\},\{c,d\}$ cross in a common point, we can \emph{planarize the crossing} by removing edges~$\{a,b\}$ and~$\{c,d\}$, introducing a new copy of~$G_\times$ at the position of the crossing, and adding the edges~$\{a,v\},\{v',b\},\{c,u\},\{u',d\}$. Garey et al.~\cite{GareyJS76} analyzed the size of a maximum independent set in~$G_\times$ when restricting which terminal vertices may occur in the set, as shown in Figure~\ref{tab:crossover}.  As~$G_\times$ is symmetric in both the horizontal and vertical axis, and the table shows that a maximum size independent set size of nine can already be obtained using~$i=1$ of the terminals~$\{v,v'\}$ and~$j=1$ of the terminals~$\{u,u'\}$, we observe the following.

\begin{observation} \label{obs:crossover}
For any choice of terminals~$v^* \in \{v,v'\}$ and~$u^* \in \{u, u'\}$ there is a maximum independent set of size nine in~$G_\times$ that does not contain~$v^*$ or~$u^*$.
\end{observation}

The following proposition summarizes the essential features of a planarization operation.

\begin{proposition} \label{prop:crossover:pties}
Let~$G$ be a graph drawn in the plane such that no edge contains a vertex in its interior and no more than two edges cross in any single point. Let~$G'$ be the result of planarizing an edge crossing by a crossover gadget. The following holds.
\begin{enumerate}
	\item For every independent set~$X$ in~$G$ there is an independent set~$X'$ in~$G'$ of size~$|X| + 9$ such that~$X' \cap V(G) = X$.\label{crossover:xtoprime}
	\item For every independent set~$X'$ in~$G'$ there is an independent set~$X''$ in~$G'$ with~$|X'| = |X''|$ containing exactly nine vertices from~$G_\times$ with~$X'' \cap V(G) \subseteq X' \cap V(G)$.\label{crossover:moveintocrossing}
	\item For every independent set~$X'$ in~$G'$ there is an independent set~$X$ in~$G$ of size~$|X'| - 9$ such that~$X \subseteq X' \cap V(G)$.\label{crossover:primetox}
	\item $\optis(G') = \optis(G) + 9$.\label{crossover:opt}
\end{enumerate}
\end{proposition}
\begin{proof}
Assume without loss of generality that the crossover gadget was introduced into~$G'$ to eliminate the crossing of edges~$\{a,b\}$ and~$\{c,d\}$.

\textbf{(\ref{crossover:xtoprime})} Since~$X$ is independent in~$G$ it contains at most one vertex from~$\{a,b\}$ and at most one vertex from~$\{c,d\}$, implying that at most one vertex~$v^* \in \{v,v'\}$ has a $G'$-neighbor in~$X$, and at most one vertex~$u^* \in \{u,u'\}$ has a $G'$-neighbor in~$X$. Let~$X_\times$ be a maximum independent set in~$G_\times$ that does not contain~$v^*$ or~$u^*$, which exists by Observation~\ref{obs:crossover}. Then~$X \cup X_\times$ is independent in~$G'$ and has size~$|X|+9$. Note that~$X' \cap V(G) = X$.

\textbf{(\ref{crossover:moveintocrossing})} Consider an arbitrary independent set~$X'$ in~$G'$. If~$X'$ contains nine vertices from~$G_\times$ we are done, so assume this is not the case. If~$X'$ contains at most seven vertices from~$G_\times$, then let~$X_\times$ be an independent set in~$G_\times$ of size nine using only the terminal vertices~$v$ and~$u$ (which exists by Observation~\ref{obs:crossover}), and define~$X'' := (X' \setminus (V(G_\times) \cup \{a,c\})) \cup X_\times$. Since we gain at least two vertices within the crossover gadget, we compensate for the loss of two vertices from~$X' \cap V(G)$ and therefore~$|X''| \geq |X'|$. It is easy to verify that~$X''$ is independent, and~$X'' \cap V(G) = (X' \cap V(G)) \setminus \{a,c\} \subseteq X' \cap V(G)$.

If~$X'$ contains eight vertices from~$G_\times$, then the table in Figure~\ref{tab:crossover} shows that~$X'$ contains at least one terminal vertex of~$G_\times$ and we can do a similar exchange. Let~$x$ be a terminal of~$G_\times$ contained in~$X'$ and let~$y$ be a terminal from~$G_\times$ that is not paired up with~$x$. By Observation~\ref{obs:crossover} there is an independent set~$X_\times$ in~$G_\times$ of size nine whose only terminals are~$x$ and~$y$. Now remove from~$X'$ the $G'$-neighbor that~$y$ has in~$V(G)$ (if any), remove the eight vertices from the crossover gadget, and replace them by~$X_\times$. We obtain an independent set~$X''$ of size at least~$|X'|$ containing nine vertices from~$G_\times$ with~$X'' \cap V(G) \subseteq X' \cap V(G)$.

\textbf{(\ref{crossover:primetox})} Consider an arbitrary independent set~$X'$ in~$G'$. By Property~\ref{crossover:moveintocrossing} there exists an independent set~$X''$ that is not smaller than~$X'$, uses exactly nine vertices from~$G_\times$, with~$X'' \cap V(G) \subseteq X' \cap V(G)$. The table in Figure~\ref{tab:crossover} shows that an independent set in~$G_\times$ of size nine contains at least one vertex from~$\{v,v'\}$ and at least one vertex from~$\{u,u'\}$. Since~$v$ and~$v'$ are connected in~$G'$ to~$a$ and~$b$, respectively, this implies that there is a vertex among~$a,b$ that is not contained in~$X''$. Similarly, the table shows that~$X''$ contains a vertex from~$\{u,u'\}$. Since these are connected to~$c$ and~$d$, one of these vertices is not in~$X''$. All edges of~$G$ also occur in~$G'$, except for the edges~$\{a,b\}$ and~$\{c,d\}$. As we have just argued that~$X''$ contains at most one vertex from each of these two edges, it follows that~$X'' \cap V(G)$ is an independent set in~$G$. Since~$X'' \cap V(G) \subseteq X' \cap V(G)$ and~$|X'' \cap V(G)| = |X''| - 9$, this proves the claim.

\textbf{(\ref{crossover:opt})} Property~\ref{crossover:xtoprime} implies that~$\optis(G') \geq \optis(G) + 9$. Property~\ref{crossover:primetox} implies that~$\optis(G) \geq \optis(G') - 9$, establishing equality.
\end{proof}

In most applications of crossover gadgets, the only important property is that they have a fixed effect on the optimum (Property~\ref{crossover:opt}). In our case we also have to ensure that the crossover gadgets do not disturb how the solutions intersect the boundary of the graph. Properties~\ref{crossover:xtoprime}--\ref{crossover:primetox} will be used for this purpose. Using these gadgets we present the construction.

\begin{lemma}\label{lemma:planargraph}
Let~$t$ be a positive integer and~$B := \{p_{1,1}, p_{2,1},\ldots, p_{t-1,1}, p_{t,1}\}$. For every non-constant monotone set-function~$f \colon 2^B \to \{0,1\}$ there is a planar graph~$G$ with boundary~$B$ such that~$\pw(G) \leq t + \Oh(1)$ and for every~$S \subseteq B\colon f(S) = 1$ if and only if~$\s_G(S) = \optis(G)$.
\end{lemma}
\begin{proof}
Consider a monotone set-function~$f$ and let~$\phi$ be a monotone CNF formula that represents~$f$ in the sense of Proposition~\ref{prop:monotone:cnf}. Let the clauses of~$\phi$ be~$C_1, \ldots, C_m$ such that each clause~$C_i$ is a subset of~$[t]$ giving the indices of the variables appearing in the clause. Since~$\phi$ is monotone, all variables appear positively. The number of literals in~$C_i$ is denoted~$|C_i|$.

We first construct a nonplanar graph~$G_\phi$ of small pathwidth such that for all~$S \subseteq B$ we have~$f(S) = 1$ if and only if~$\s_{G_\phi}(S) = \optis(G_\phi)$. Then we will use crossover gadgets to turn~$G_\phi$ into a planar graph~$G'_\phi$ while preserving these properties. The construction is inspired by a reduction of Lokshtanov et al.~\cite[Thm.~3.1]{LokshtanovMS11a}, and proceeds as follows.
\begin{enumerate}
	\item We start by creating $t$ paths $P_1,\ldots,P_t$, where every path $P_i$ for~$i \in [t]$ consists of $2m$ vertices $p_{i,1},\ldots,p_{i,2m}$. The boundary~$B=\{p_{1,1},\ldots,p_{t,1}\}$ of graph~$G_\phi$ contains the first vertex from each path.
	\item For each clause~$i \in [m]$, add a copy of the clause gadget~$\mathcal{C}_{|C_i|}$ to the graph and denote its terminals by~$(v_1, \ldots, v_{|C_i|})$. Let~$\ell(j)$ denote the $j$-th variable in the clause for each~$j \in [|C_i|]$ and sort these such that~$\ell(1) > \ell(2) > \ldots > \ell(|C_i|)$; this will be useful later on when planarizing the graph. For each~$j \in [|C_i|]$ make terminal~$v_j$ in the clause gadget adjacent to vertex~$p_{\ell(j), 2i}$ on path~$P_{\ell(j)}$. Observe that clause gadgets only connect to \emph{even-numbered} vertices on the paths.
\end{enumerate}

\begin{claim} \label{claim:gphi:properties}
The graph~$G_\phi$ with boundary~$B := \{p_{1,1}, \ldots, p_{t,1}\}$ satisfies:
\begin{enumerate}
			\item $\s_{G_\phi}(B) = \optis(G_\phi) \leq mt +\sum_{1\leq i\leq m} (|C_i| + 2)$. \label{gphi:smallis}
			\item $\s_{G_\phi}(B) = \optis(G_\phi) = mt +\sum_{1\leq i\leq m} (|C_i| + 2)$. \label{gphi:bigis}
			\item For each~$S \subseteq B$ we have~$f(S) = 1$ if and only if~$\s_{G_\phi}(S) = \optis(G_\phi)$.\label{gphi:representsf}
\end{enumerate}
\end{claim}
\begin{claimproof}
We prove the properties of~$G_\phi$ one by one.

\textbf{(\ref{gphi:smallis})} By Definition~\ref{def:subsetfunc} we have $s_{G_\phi}(B) = \optis(G_\phi)$. Observe that $G_\phi$ consists of $t$ paths of $2m$ vertices each, and $m$ clause gadgets. As an independent set on a path never contains two subsequent vertices in the path, any independent set in~$G_\phi$ contains at most $m$ vertices from each of the~$t$ paths. By Observation~\ref{obs:clausegadget} we know that for each clause~$C_i$ of size~$|C_i|$, an independent set contains at most~$|C_i| + 2$ vertices from the created clause gadget. Hence an independent set in~$G_\phi$ has size at most~$mt+\sum_{1\leq i\leq m} (|C_i| + 2)$.

\textbf{(\ref{gphi:bigis}}) Consider the set~$X$ defined as follows:
\begin{itemize}
	\item Set~$X$ contains the odd-numbered vertices~$p_{i,1}, \ldots, p_{i,2m-1}$ from each path~$P_i$ with~$i \in [t]$.
	\item For each clause~$i \in [m]$, let~$X_i$ be an independent set of size~$|C_i| + 2$ in the clause gadget created for~$C_i$; such a set exists by Observation~\ref{obs:clausegadget}. Add~$X_i$ to~$X$ for each clause~$i \in [m]$.
\end{itemize}
It is easy to see that~$X$ has the desired size. To see that~$X$ is an independent set, note that the vertices we choose from each path form an independent set and that there are no edges between vertices on paths~$P_i$ and~$P_j$ for~$i \neq j$. Finally, observe that the clause gadgets are only connected to the rest of the graph through terminals, and that the terminals are adjacent to even-numbered vertices on the paths. As~$X$ only contains odd-numbered vertices from the paths, it follows that~$X$ is an independent set in~$G_\phi$ and therefore that~$\optis(G_\phi) \geq mt + \sum_{1 \leq i \leq m} (|C_i| + 2)$. Together with Property~\ref{gphi:smallis} of Claim~\ref{claim:gphi:properties} this shows that~$\optis(G_\phi)$ is \emph{exactly} the given number.

\textbf{(\ref{gphi:representsf})} We show that~$f(S) = 1$ if and only if there is a maximum independent set in~$G_\phi$ whose intersection with the boundary is a subset of~$S$. Since~$\phi$ represents~$f$, it is sufficient to argue that a $0/1$-assignment to variables~$x_1, \ldots, x_t$ satisfies the CNF~$\phi$ if and only if there is a maximum independent set~$X$ in~$G_\phi$ with~$X \cap B \subseteq S$ for the set~$S := \{ p_{i, 1} \mid i \in [t], x_i = 1 \}$. We will prove the two directions of this equivalence separately.

\textbf{(\ref{gphi:representsf}, $\Rightarrow$)} Assume we have an assignment of variables $x_1,\ldots,x_t\in\{0,1\}$ such that $\phi(x_1,\ldots,x_t) = 1$. This means that every clause $C_i$ has at least one positive literal~$\ell(j)$ that is set to~$1$. We construct an independent set of size~$mt + \sum_{i=1}^m (|C_i| + 2)$ in~$G_\phi$ which only contains boundary vertices whose corresponding variable is set to~$1$, as follows.
\begin{itemize}
	\item For each~$i \in [t]$ the set~$X$ contains the odd-numbered vertices~$p_{i,1}, \ldots, p_{i, 2m-1}$ if~$x_i = 1$, and it contains the even-numbered vertices~$p_{i,2}, \ldots, p_{i,2m}$ if~$x_i = 0$. This ensures that~$X$ indeed only contains boundary vertices whose variable is set to~$1$.
	\item For each clause~$i \in [m]$, some variable~$\ell(j)$ appearing in clause~$C_i$ is set to~$1$ since the assignment satisfies~$\phi$. Let~$X_i$ be an independent set of size~$|C_i| + 2$ in the clause gadget created for~$C_i$, such that~$v_{\ell(j)}$ is the only terminal vertex contained in~$X_i$. Such a set exists by Property~\ref{clausegadget:oneterminal} of Observation~\ref{obs:clausegadget}. Add~$X_i$ to~$X$ for each clause~$i \in [m]$.
\end{itemize}
The constructed set~$X$ has the claimed size, which shows it has the size of a maximum independent set by Property~\ref{gphi:bigis} of Claim~\ref{claim:gphi:properties}. It remains to prove that~$X$ is an independent set. The only nontrivial part is to show that the vertices chosen in the clause gadgets are not adjacent to those chosen on the paths. Consider an arbitrary clause~$C_i$. The set~$X_i$ we added to~$X$ for this clause contains only one terminal vertex, for a variable~$v_{\ell(j)}$ that was set to~$1$. By our construction of~$X$, this implies~$X$ contains the odd-numbered vertices on path~$P_{\ell(j)}$. Since terminal vertices are only adjacent to even-numbered vertices of the path, no vertex of~$X_i$ is adjacent to a path-vertex in~$X$, showing that~$X$ is independent.

\textbf{(\ref{gphi:representsf}, $\Leftarrow$)} Consider an arbitrary set~$S\subseteq B$ for which~$\s_{G_\phi}(S) = \optis(G_\phi)$, which equals~$mt+\sum_{1\leq i\leq m} (|C_i| + 2)$ by Property~\ref{gphi:bigis} of Claim~\ref{claim:gphi:properties}. By Definition~\ref{def:subsetfunc}, we know that there is an independent set $X$ of size $mt+\sum_{1\leq i\leq m} (|C_i| + 2)$ in $G_\phi$ such that $X\cap B \subseteq S$.

An independent set contains at most~$m$ vertices from each of the~$t$ paths, and contains at most~$|C_i| + 2$ vertices from each gadget constructed for a clause~$i \in [m]$ by Observation~\ref{obs:clausegadget}. It follows that to attain its claimed size,~$X$ must contain \emph{exactly}~$m$ vertices from each path and \emph{exactly}~$|C_i|+2$ vertices from each clause gadget~$i \in [m]$. By Property~\ref{clausegadget:chooses:terminal} of Observation~\ref{obs:clausegadget} it follows that~$X$ contains a terminal vertex from each clause gadget. 

Suppose that~$v_{\ell(j)}$ is a terminal vertex of the gadget for clause~$C_i$ and~$v_{\ell(j)}$ belongs to~$X$. Since~$v_{\ell(j)}$ is adjacent to~$p_{\ell(j), 2i}$, the latter vertex is not contained in~$X$. Since~$X$ contains~$m$ vertices from the path~$P_{\ell(j)}$, it contains exactly one vertex from each pair~$\{p_{\ell(j), 2k-1}, p_{\ell(j), 2k}\}$ for~$1 \leq k \leq m$. As~$X$ does not contain the even-numbered vertex~$p_{\ell(j), 2i}$, it contains the odd-numbered vertex~$p_{\ell(j), 2i-1}$ from the pair~$\{p_{\ell(j), 2i-1}, p_{\ell(j), 2i}\}$. Since the vertices are connected in a path, this propagates to lower indices and ensures that from all pairs with~$1 \leq k \leq i$ the set~$X$ contains the odd-numbered vertex. In particular, this shows that~$X$ contains~$p_{\ell(j), 1}$. Since this is a boundary vertex, and~$X$ intersects the boundary in a subset of~$S$, it follows that~$p_{\ell(j), 1} \in S$.

The argument given above shows that each clause of~$\phi$ contains a variable whose corresponding boundary vertex is contained in~$S$. Since~$\phi$ has no negated literals, it follows that~$\phi$ is satisfied by setting the variables corresponding to the vertices in~$S$ to~$1$. As~$\phi$ represents the monotone Boolean function~$f$, we find that~$f$ outputs~$1$ on~$S$.
\end{claimproof}

\begin{figure}[t]
\centering
\includegraphics[width=.9\textwidth]{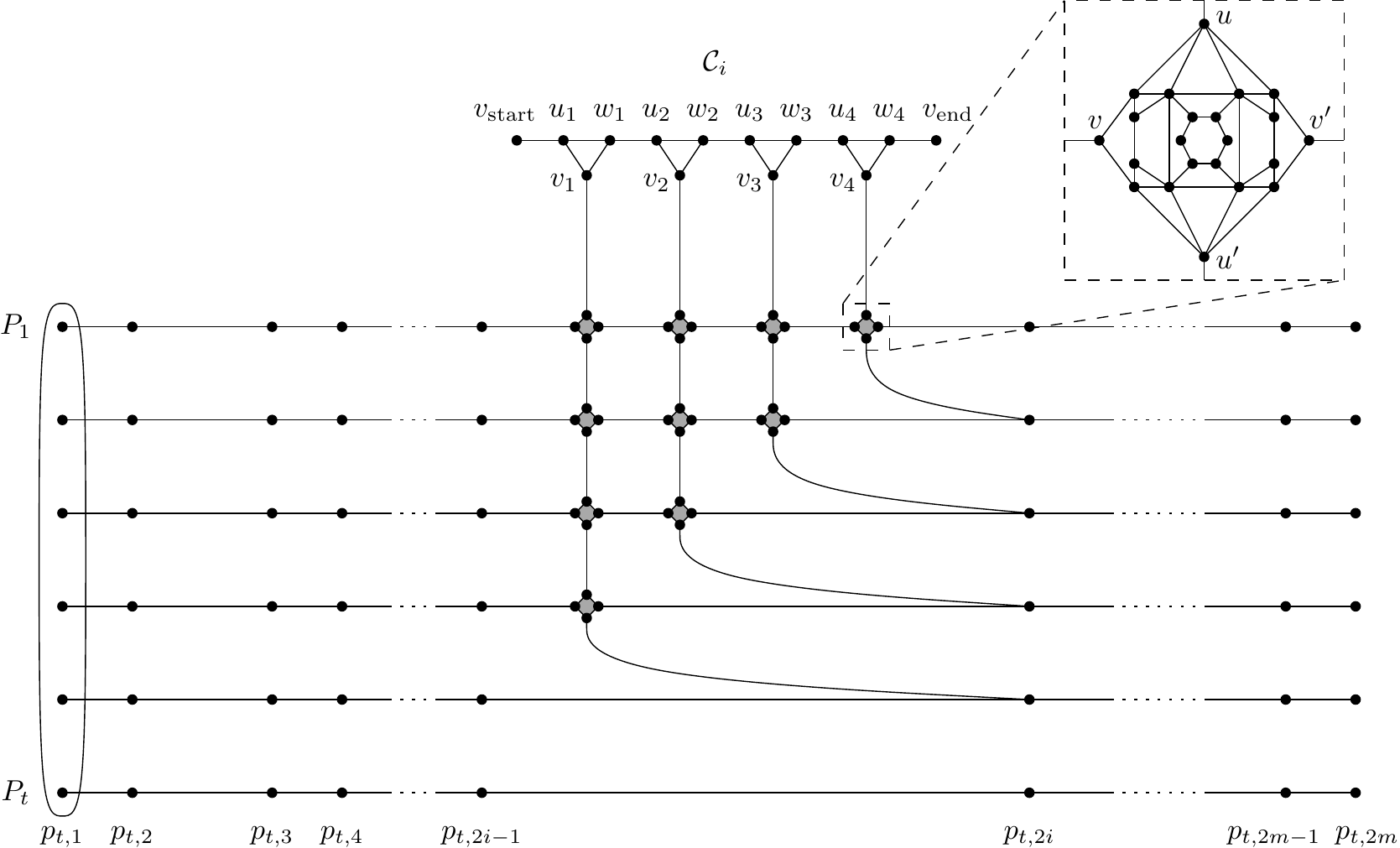}
\caption{Planarizing the graph~$G_\phi$ to obtain~$G'_{\phi}$ in the proof of Lemma~\ref{lemma:planargraph}. Only the clause gadget for the clause~$C_i = (x_5 \vee x_4 \vee x_3 \vee x_2)$ is shown. Shaded diamonds represent crossover gadgets. The boundary~$B$ of the graph is circled, containing the first vertex from each path.}
\label{fig:planarizing}
\end{figure}

Claim~\ref{claim:gphi:properties} shows that the boundary function of~$G_\phi$ expresses the monotone Boolean function~$f$. The same argumentation as used by Lokshtanov et al.~\cite[Lemma 3.3]{LokshtanovMS11a} shows that~$G_\phi$ has pathwidth~$t + \Oh(1)$. However, we will not prove this here for the non-planar graph~$G_\phi$; we will prove a pathwidth bound after planarizing the graph. The planarization starts from a drawing of~$G_\phi$ in the plane in which the crossings have a fixed structure. This drawing is defined as follows (see Figure~\ref{fig:planarizing}):

\begin{itemize}
	\item Draw each path~$P_1, \ldots, P_t$ horizontally. Place the paths above each other so that~$P_1$ is the highest and~$P_t$ is the lowest.
	\item For each clause~$i \in [m]$ of~$\phi$, draw the clause gadget in a planar fashion above the paths, so that its terminals stick out at the bottom, the lowest-indexed terminal on the left and the highest-numbered terminal on the right. Draw the gadget for clause~$i$ between the vertical lines containing the $2i-1$-th and the $2i$-th vertices on each path. Consider the set of edges $E_{C_i}$ connecting the gadget for clause~$C_i$ to the vertices of the paths. By construction of~$G_\phi$, the gadget only connects to vertices with index~$2i$ on the paths. Draw the edges from~$E_{C_i}$ in such a way that $e \in E_{C_i}$ only crosses the edges between the vertices $p_{j,2i-1}$ and $p_{j,2i}$ of the paths $P_j$ for~$j \in [t]$, and do not cross any other edge $e' \in E_{C_i}$. Since the left-to-right order of the variables in a clause matches the order in which the paths are laid out from top to bottom, this is possible.
\end{itemize}

Based on this drawing we planarize the graph~$G_\phi$ by repeatedly replacing crossings by crossover gadgets, resulting in a planar graph~$G'_\phi$ as shown in Figure~\ref{fig:planarizing}. Let~$N$ denote the number of crossover gadgets which were introduced during the planarization process. By Proposition~\ref{prop:crossover:pties} we know that~$\optis(G'_\phi) = \optis(G_\phi) + 9N = mt + 9N +\sum_{1\leq i\leq m} (|C_i| + 2)$, where we use Property~\ref{gphi:bigis} of Claim~\ref{claim:gphi:properties} for the second equality. To conclude the proof, it remains to show that~$\pw(G'_\phi) \leq t + \Oh(1)$ (Claim~\ref{claim:phiprime:pathwidth}) and that for all subsets~$S \subseteq B$ we have~$f(S) = 1$ if and only if~$\s_{G'_\phi}(S) = \optis(G'_\phi)$ (Claim~\ref{claim:phiprime:representsf}).

\begin{claim} \label{claim:phiprime:representsf}
For every~$S \subseteq B$ we have~$f(S) = 1$ if and only if~$\s_{G'_\phi}(S) = \optis(G'_\phi)$.
\end{claim}
\begin{claimproof}
\textbf{($\Rightarrow$)} Suppose that~$f(S) = 1$ for~$S \subseteq B$. By Property~\ref{gphi:representsf} of Claim~\ref{claim:gphi:properties}, we have~$\s_{G_\phi}(S) = \optis(G_\phi)$, implying~$G_\phi$ has a maximum independent~$X$ with~$X \cap B \subseteq S$. By repeated application of Property~\ref{crossover:xtoprime} of Proposition~\ref{prop:crossover:pties}, this implies that~$G'_\phi$ has an independent set~$X'$ of size~$|X| + 9N = \optis(G'_\phi)$ with~$X' \cap V(G_\phi) = X \cap V(G_\phi)$, implying in particular that~$X' \cap B \subseteq S$.

\textbf{($\Leftarrow$)} Suppose that~$\s_{G'_\phi}(S) = \optis(G'_\phi)$. Using Definition~\ref{def:subsetfunc} this implies that~$G'_\phi$ has a maximum independent set~$X'$ with~$X' \cap B \subseteq S$. By repeated application of Property~\ref{crossover:primetox} of Proposition~\ref{prop:crossover:pties} this implies that~$G$ has an independent set~$X$ of size~$\optis(G'_\phi) - 9N = \optis(G_\phi)$ such that~$X \subseteq X' \cap V(G_\phi)$. Hence~$X \cap B \subseteq X' \cap V(G_\phi) \cap B = X' \cap B \subseteq S$. Hence~$G_\phi$ has a maximum independent set intersecting~$B$ in a subset of~$S$. By Property~\ref{gphi:representsf} of Claim~\ref{claim:gphi:properties} this implies that~$f(S) = 1$.
\end{claimproof}

\begin{claim} \label{claim:phiprime:pathwidth}
The graph~$G'_\phi$ has pathwidth~$t + \Oh(1)$.
\end{claim}
\begin{claimproof}
To bound the pathwidth of~$G'_\phi$ we use a \emph{Mixed Search Game}~\cite{takahashi1995mixed}, following Lokshtanov et al.~\cite[Lemma 3.3]{LokshtanovMS11a}. We interpret the graph as a network of contaminated tunnels. Initially, all the edges are contaminated; the goal is to clean the edges by cleaners. Cleaners reside at vertices and slide along edges. An edge is cleaned when both endpoints contain a cleaner, or when a cleaner slides along the edge. Clean edges are recontaminated when there is a path without cleaners from a contaminated edge to a clean edge. It is known that the minimum number of cleaners needed to clean the graph is an upper bound on its pathwidth~\cite{takahashi1995mixed}. We present a strategy to clean~$G'_\phi$ using~$t + \Oh(1)$ cleaners. The cleaning process proceeds in \emph{rounds} corresponding to the clauses of~$\phi$. For~$i \in [m]$, the $i$-th round starts in this situation:
\begin{itemize}
	\item There are cleaners on all~$t$ vertices~$p_{j, 2i-1}$ for~$j \in [t]$.
	\item All edges incident on clause gadgets for clauses~$i' < i$ are clean.
	\item The edges incident on path vertices~$p_{j, i'}$ for~$j \in [t]$ and~$i' < 2i-1$ are clean.
	\item The crossover gadgets that were introduced to eliminate crossings of edges incident on the gadgets for clauses~$i' < i$ are clean.
\end{itemize}
Note that these preconditions trivially hold at the beginning of the first round for~$i=1$. A round for clause~$i \in [m]$ starts by placing a cleaner on the \vstart vertex of the gadget for clause~$i$ and sliding it into~$u_1$. Then we consider each path in turn, from bottom to top. When the cleaner on the path~$P_j$ is already at its target position for this round (vertex~$p_{j,2i}$) then we leave it there. Otherwise, we slide the cleaner over the edge to its right. If the arrival vertex is a terminal of a crossover gadget, then we temporarily place three additional cleaners on the other terminals of the crossover gadget, and clean the interior of the crossover gadget (which can be done using at most four temporary cleaners). Afterward we remove the cleaners from the south and west terminals of the crossover gadget, and continue with the path one higher. After the crossover gadget for the higher path has also been cleaned, we can remove the cleaner from the north terminal of the lower path. In this way we propagate upwards, moving the cleaner on each path beyond the next crossover gadget. When we reach the topmost path we use one temporary cleaner in the clause gadget, together with the cleaner that was already in the clause gadget, to clean the next triangle on the crossover gadget. Using Figure~\ref{fig:cleaning} it is straightforward to verify that by repeating this process for all literals in the clause, we can clean the entire clause gadget and move the cleaners onto the starting position for the next round. After the clause gadget is fully cleaned, we remove the cleaners from it before starting the next round.

\begin{figure}[t]
\centering
\includegraphics[width=.9\textwidth]{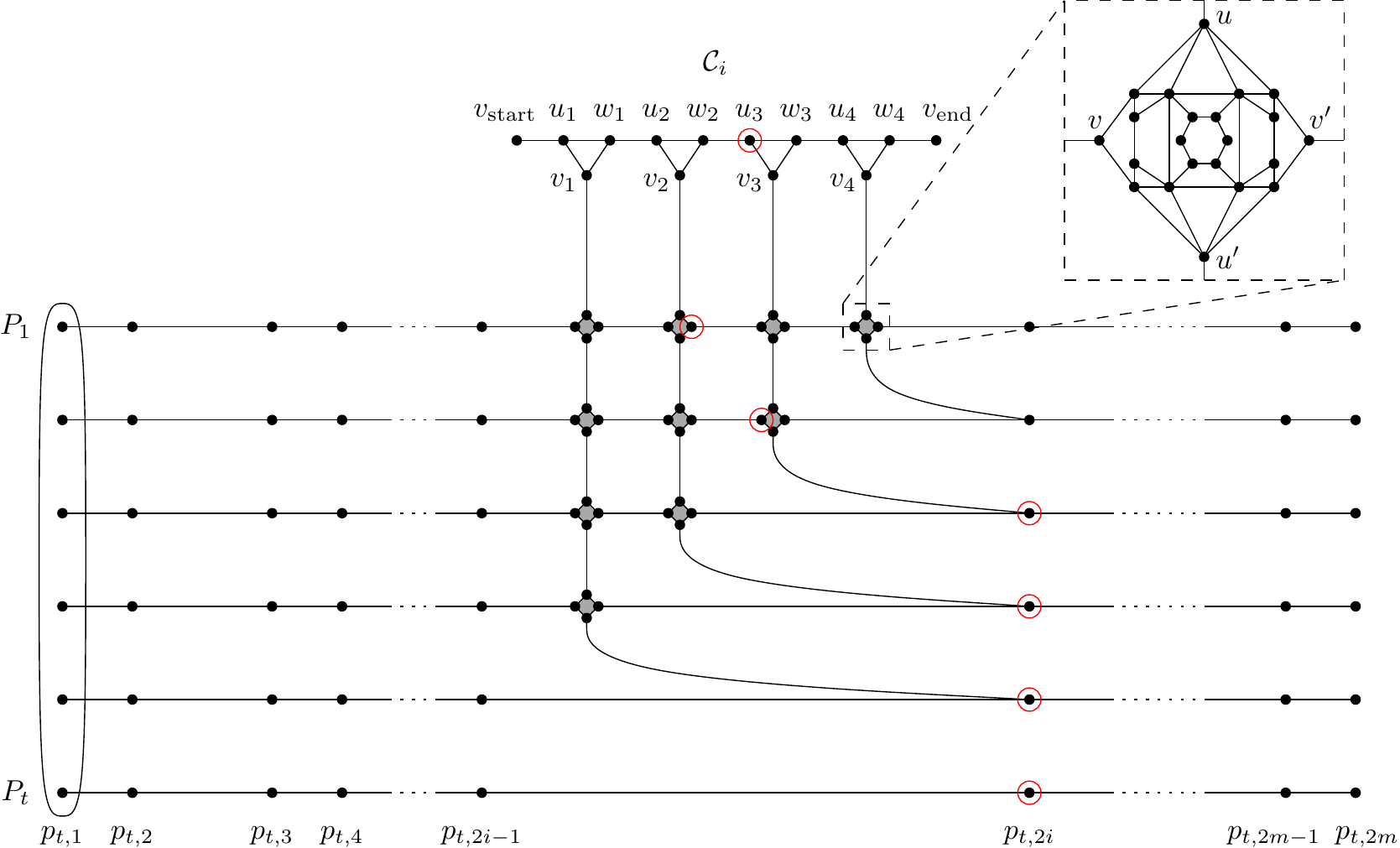}
\caption{Sketch of the cleaning process for the graph~$G'_\phi$. Cleaners are located at the vertices highlighted by hollow (red) circles. The left part of the graph is already cleaned. Only the clause gadget for the clause~$C_i = (x_5 \vee x_4 \vee x_3 \vee x_2)$ is shown. Shaded diamonds represent crossover gadgets.}
\label{fig:cleaning}
\end{figure}

To see that we never need more than~$t + \Oh(1)$ cleaners simultaneously, observe that:
\begin{itemize}
	\item We never have cleaners on more than one clause gadget simultaneously.
	\item The number of cleaners needed in a clause gadget, or in any single crossover gadget, is constant.
	\item At any point in the process, there are at most two crossover gadgets from which cleaners occupy more than one terminal. These are the two crossover gadgets where the cleaning process propagates upwards.
	\item In an idle state, there is exactly one cleaner on each of the paths~$P_j$ for~$j \in [t]$.
\end{itemize}
It follows that~$G'_\phi$ can be cleaned with~$t + \Oh(1)$ cleaners and therefore has pathwidth at most~$t + \Oh(1)$. A more careful analysis shows that~$\pw(G'_\phi) \leq t+6$.
\end{claimproof}

This concludes the proof of Lemma~\ref{lemma:planargraph}.
\end{proof}

\section{Lower bound for protrusion replacement} \label{sec:lowerbounds:protrusions}
To leverage the construction of Lemma~\ref{lemma:planargraph} into a lower bound on the critical size of a set of representatives, we need the following lemma. Observe that its second condition shows that no pair of graphs from the constructed set~$\mathcal{G}$ is equivalent under~$\equiv_{\IS,t}$, and this is witnessed already by gluing planar graphs of pathwidth one onto them. This implies that in any protrusion reduction scheme applied to planar graphs that aims to replace occurrences of bounded-pathwidth protrusions by representatives, there should be a distinct representative for each graph in~$\mathcal{G}$.

\begin{lemma} \label{lemma:nonequivalent:planar}
For each positive integer~$t$ there is a set~$\mathcal{G}$ of~$M(t) - 2$ distinct $t$-boundaried planar graphs of pathwidth~$t + \Oh(1)$, such that for each pair of distinct graphs~$G_f,G_{f'} \in \mathcal{G}$ there are two indicator graphs~$I_{S}$ and~$I_B$ as in Definition~\ref{def:indicator} such that:
\begin{enumerate}
	\item The graphs~$G_f \oplus I_S, G_f \oplus I_B, G_{f'} \oplus I_S$, $G_{f'} \oplus I_B$ are planar and have pathwidth~$t + \Oh(1)$.
	\item $\optis(G_f \oplus I_S) - \optis(G_{f'} \oplus I_S) \neq \optis(G_f \oplus I_B) - \optis(G_{f'} \oplus I_B)$.\label{noneqv:planargraphs}
\end{enumerate}
\end{lemma}
\begin{proof}
Recall that~$M(t)$ counts the number of distinct monotone Boolean functions, which we will interpret as monotone set-functions of the form~$2^B \to \{0,1\}$ for~$B := \{p_{1,1}, p_{2,1},\ldots, \linebreak[1] p_{t-1,1}, p_{t,1}\}$. Since there is one such function that is constantly~$0$, and one that is constantly~$1$, there is a set~$\mathcal{F}$ of~$M(t)-2$ distinct non-constant monotone Boolean functions. For each function~$f \in \mathcal{F}$, apply Lemma~\ref{lemma:planargraph} to obtain a $t$-boundaried planar graph~$G_f$ with boundary~$B$ and let~$\mathcal{G}$ be the resulting set of boundaried graphs. We prove that~$\mathcal{G}$ has the two claimed properties.	

Consider two distinct boundaried graphs~$G_{f}$ and~$G_{f'}$ in~$\mathcal{G}$. Since~$f$ and~$f'$ are monotone and not constantly~$0$, we have~$f(B) = f'(B) = 1$. Since~$f$ and~$f'$ are distinct, there is a set~$S \subseteq B$ such that~$f(S) \neq f'(S)$, implying that~$S \neq B$. Consider the indicator graphs~$I_B$ and~$I_S$ as in Definition~\ref{def:indicator}. We claim that gluing either of these graphs to~$G_f$ or~$G_{f'}$ results in a planar graph of pathwidth~$t + \Oh(1)$. To see this, observe that indicator graphs are disjoint unions of stars; gluing stars onto a graph does not violate planarity. The mixed search strategy of Claim~\ref{claim:phiprime:pathwidth} can be adapted to work without extra cleaners when stars are glued onto the boundary: before starting the cleaning, place cleaners on all~$t$ boundary vertices and jump one cleaner over all leaves of the glued stars to clean the new edges. Afterward, the cleaning of the remainder of the graph proceeds as before. It remains to prove the last part of the lemma statement.

Since~$f(S) \neq f(S')$ one of these values is~$0$ and the other is~$1$; assume without loss of generality that~$f(S) = 1$ and~$f'(S) = 0$. We combine Proposition~\ref{prop:glueindicator} with the guarantees of Lemma~\ref{lemma:planargraph} to prove the following.

\begin{claim} \label{claim:nonequivalence:planarglued}
The following holds.
\begin{enumerate}
	\item $\optis(G_f \oplus I_B) = \optis(G_f)$.
	\item $\optis(G_f \oplus I_S) = \optis(G_f) + 2(t - |S|)$.
	\item $\optis(G_{f'} \oplus I_B) = \optis(G_{f'})$.
	\item $\optis(G_{f'} \oplus I_S) = \optis(G_{f'}) + 2(t - |S|) + c$ for some~$c \neq 0$.
\end{enumerate}
\end{claim}
\begin{claimproof}
By Proposition~\ref{prop:glueindicator} we have~$\optis(G_f \oplus I_B) = \s_{G_f}(B) + 2(t - |B|) = \s_{G_f}(B) = \optis(G_f)$, where the last step follows from Definition~\ref{def:subsetfunc}. Similarly, we have~$\optis(G_{f'} \oplus I_B) = \optis(G'_f)$.

Now we consider the effect of gluing~$I_S$. By Lemma~\ref{lemma:planargraph}, the fact that~$f(S) = 1$ implies that~$\s_{G_f}(S) = \optis(G_f)$. Since Proposition~\ref{prop:glueindicator} states that~$\optis(G_f \oplus I_S) = \s_{G_f}(S) + 2(t-|S|)$ this yields~$\optis(G_f \oplus I_S) = \optis(G_f) + 2(t - |S|)$. Since~$f'(S) = 0$, Lemma~\ref{lemma:planargraph} implies~$\s_{G'_f}(S) \neq \optis(G_{f'})$, so~$\s_{G'_f}(S) = \optis(G_{f'}) + c$ for some~$c \neq 0$. A final application of Proposition~\ref{prop:glueindicator} gives~$\optis(G'_f \oplus I_S) = \optis(G_{f'}) + c + 2(t - |S|)$.
\end{claimproof}

The last part of the lemma follows directly from plugging in the derived values of~$\optis$.
\end{proof}

Finally, we can combine our lower bound on the number of distinct equivalence classes of~$\equiv_{\IS,t}$ (Lemma~\ref{lemma:nonequivalent:planar}) with an upper bound on the number of small graphs to obtain a lower bound for protrusion replacement for \IndependentSet.

\begin{theorem} \label{theorem:planar:protrusion:lowerbound}
Let~$t \geq t_0$ be a sufficiently large positive integer. Let~$\mathcal{R}_t$ be a set of $t$-boundaried planar graphs such that every equivalence class of~$\equiv_{\IS,t}$ that contains a planar graph of pathwidth~$t + \Oh(1)$ is represented by some graph in~$\mathcal{R}_t$. Then~$\mathcal{R}_t$ contains a graph with~$\Omega(\log M(t)) \geq \Omega(2^t / \sqrt{4t})$ vertices.
\end{theorem}
\begin{proof}
The core idea is that there are not enough distinct small $t$-boundaried planar graphs to represent each of the~$M(t)-2$ distinct equivalence classes identified in Lemma~\ref{lemma:nonequivalent:planar} by a different small graph; hence at least one representative must be large. To make this explicit, let~$N_t(n)$ denote the number of distinct $n$-vertex $t$-boundaried planar graphs for~$n \in \mathbb{N}$. If there are fewer than $M(t)-2$ distinct $t$-boundaried planar graphs of size at most~$x$, then some representative has size at least~$x+1$. This means that, to prove the theorem, we have to show that there are constants $\alpha,t_0$, such that for every $t \geq t_0$ it holds that~$\sum_{n=t}^{\lfloor \alpha\log{M(t) \rfloor}}{N_t(n)} < M(t)-2$. (There can be no $t$-boundaried graph with fewer than~$t$ vertices.)
	
Bonichon et al.~\cite{BonichonGHPS06} showed that for all~$n \in \mathbb{N}$, there are fewer than~$2^{4.91n}$ distinct unlabeled $n$-vertex planar graphs. For our application, we need to count planar graphs that have a labeled boundary of exactly~$t$ vertices. Since there are~$\binom{n}{t}$ ways to choose a distinct boundary of size~$t$ in a given unlabeled $n$-vertex graph, we have:
$$N_n(t) < 2^{4.91n} \cdot \binom{n}{t} \leq 2^{4.91n} \cdot 2^n = 2^{5.91n}.$$ 
Recall that for~$x \in \mathbb{R}$ with~$x \neq 1$ the geometric series~$\sum _{i=0}^n x^k$ equals~$\frac{x^{n+1}-1}{x-1}$. 
\allowdisplaybreaks
\begin{align*}
	\sum_{n=t}^{\lfloor \alpha\log{M(t) \rfloor}} N_t(n) &\leq \sum_{n=t}^{\lfloor \alpha\log{M(t)} \rfloor}{2^{5.91n}}  & \mbox{Stated bounds on~$N_t(n)$.}  \\
	& < 2^{5.91 (\alpha \log M(t) + 1)}/(2^{5.91} - 1) & \mbox{Geometric series.} \\
	&= (2^{5.91} \cdot (2^{\log M(t)})^{5.91 \alpha})/(2^{5.91} - 1) & \mbox{Rewriting exponent.} \\
	&= (2^{5.91} \cdot M(t)^{5.91 \alpha})/(2^{5.91} - 1) & \mbox{$2^{\log_2 M(t)} = M(t)$.} \\
\end{align*}

It follows that for any~$\alpha < \frac{1}{5.91}$, the previous sum is bounded by~$\Oh(M(t)^{1 - \varepsilon})$ for some~$\varepsilon > 0$. For any such~$\alpha$ there therefore exists some~$t_0$ such that for all~$t \geq t_0$, the number of $t$-boundaried planar graphs with at most~$\lfloor \alpha \log M(t) \rfloor$ vertices is smaller than the number of equivalence classes that contain a planar graph of pathwidth~$t + \Oh(1)$, which is at least~$M(t) - 2$.
\end{proof}

For concrete, small values of~$t$, one can use exact values for~$M(t)$ (sequence A000372 in OEIS) and the number of unlabeled planar graphs (sequence A005470 in OEIS) to obtain explicit lower bounds from the counting argument in the proof of Theorem~\ref{theorem:planar:protrusion:lowerbound}. For example, for a boundary of size six we have~$M(6) = 7828354$. To have at least~$M(6) - 2$ distinct $t$-boundaried planar graphs, one requires $t$-boundaried planar graphs on at least ten vertices.

\section{Extending the lower bound to Dominating Set} \label{sec:domset}

In this section we use a folklore reduction from \VertexCover (the dual of \IndependentSet) to \DominatingSet, to obtain lower bounds for the latter problem. For a graph~$G$, let~$\ett(G)$ denote the graph obtained from~$G$ as follows: for each edge~$\{u,v\} \in E(G)$, insert a new vertex~$x_{\{u,v\}}$ that is adjacent to~$u$ and~$v$. In this way, for every edge of~$G$ we introduce a triangle with a new private vertex into~$G'$. We extend the definition of~$\ett$ to $t$-boundaried graphs; the set of boundary vertices remains unchanged by the transformation.

\begin{observation} \label{obs:ett:planar}
If~$G$ is planar, then~$\ett(G)$ is planar: the added degree-two vertices can be drawn alongside the edge for which they are inserted.
\end{observation}

\begin{observation} \label{obs:glue:ett}
If~$G_1$ and~$G_2$ are $t$-boundaried graphs in which the boundary is an independent set, then~$\ett(G_1 \oplus G_2) = \ett(G_1) \oplus \ett(G_2)$.
\end{observation}

\begin{proposition} \label{prop:ett:pathwidth}
For any graph~$G$ we have~$\pw(\ett(G)) \leq \pw(G) + 1$.
\end{proposition}
\begin{proof}
To prove that the pathwidth increases by at most one (which is unavoidable in some cases), we use the definition of pathwidth in terms of path decompositions. Let~$\mathcal{P} = X_1, \ldots, X_\ell$ be a sequence of subsets of~$V(G)$ called \emph{bags} that form a path decomposition of~$G$: each vertex occurs in a contiguous interval of bags, each vertex occurs in at least one bag, and for each edge there is a bag containing both its endpoints. The \emph{width} of the decomposition is the size of the largest bag minus one. 

Let~$X$ be the vertices of~$\ett(G)$ that do not occur in~$G$, which were inserted to complete edges into triangles. To obtain a path decomposition of~$\ett(G)$, start by replacing each bag of~$\mathcal{P}$ by~$|X|$ consecutive copies of that same bag; this does not increase the width. For each vertex~$x_{\{u,v\}} \in X$ there is an edge~$\{u,v\} \in E(G)$, and therefore~$\mathcal{P}$ has a bag containing both~$u$ and~$v$. Since we duplicated each bag~$|X|$ times, for each vertex~$x_{\{u,v\}}$ in~$X$ we can find a distinct bag containing both its endpoints. Insert each vertex of~$X$ into its associated bag; it is easy to verify that we obtain a path decomposition of~$\ett(G)$. Since the maximum size of a bag increases by exactly one, the width of the new decomposition is one larger than the width of~$\mathcal{P}$. Since~$G$ has a decomposition of pathwidth~$\pw(G)$ by definition, this concludes the proof.
\end{proof}

The following proposition is folklore; it is the core of the classic reduction from \VertexCover to \DominatingSet.

\begin{proposition} \label{prop:vc:to:ds}
Let~$G$ be a graph without isolated vertices. The following holds:~$\optvc(G) = \optds(\ett(G))$.
\end{proposition}
\begin{proof}
If~$S$ is a vertex cover in~$G$, then~$S$ is also a dominating set in~$\ett(G)$. To see that, note that~$S$ dominates all vertices of the form~$x_{\{u,v\}}$ that were inserted into~$\ett(G)$ on account of an edge~$\{u,v\} \in E(G)$, since~$S$ includes at least one of~$u$ and~$v$ to cover that edge. To see that~$S$ also dominates all original vertices, let~$v \in V(G) \cap V(\ett(G))$ be an arbitrary vertex and let~$u \in V(G)$ be an arbitrary neighbor of~$v$, which exists since~$G$ has no isolated vertices. Then~$S$ contains one of~$u$ and~$v$ to cover edge~$\{u,v\}$ and therefore dominates~$v$. It follows that~$S$ is a dominating set in~$\ett(G)$, implying that~$\optvc(G) \geq \optds(\ett(G))$.

For the other direction, consider a dominating set~$S$ in~$\ett(G)$. We may assume without loss of generality that~$S$ contains no vertices of the form~$x_{\{u,v\}}$: since the closed neighborhood of~$x_{\{u,v\}}$ is a subset of the closed neighborhoods of~$u$ and~$v$, an occurrence of a vertex~$x_{\{u,v\}}$ in~$S$ may be replaced by either~$u$ or~$v$ to obtain a new dominating set that is not larger. Hence~$\ett(G)$ has a minimum dominating set~$S$ consisting only of original vertices from~$V(G)$. For each edge~$\{u,v\} \in E(G)$, the set~$S$ contains one of~$u$ or~$v$ to dominate the vertex~$x_{\{u,v\}}$. Hence~$S$ is a vertex cover in~$G$, showing that~$\optvc(G) \leq \optds(\ett(G))$. This establishes equality and concludes the proof.
\end{proof}

Proposition~\ref{prop:vc:to:ds} allows us to transform the set of nonequivalent graphs for~$\equiv_\IS$ constructed in Lemma~\ref{lemma:nonequivalent:planar}, into a set of nonequivalent graphs for~$\equiv_\DS$. 

\begin{lemma} \label{lemma:nonequiv:ds}
For each positive integer~$t$ there is a set~$\mathcal{G}_\DS$ of~$M(t) - 2$ distinct $t$-boundaried planar graphs of pathwidth~$t + \Oh(1)$, such that for each pair of distinct graphs~$G_{\DS,1},G_{\DS,2} \in \mathcal{G}_\DS$ there are two indicator graphs~$I_1$ and~$I_2$ as in Definition~\ref{def:indicator} such that:
\begin{enumerate}
	\item For all~$i,j \in \{1,2\}$ the graph~$G_{\DS,i} \oplus \ett(I_j)$ is planar and has pathwidth~$t + \Oh(1)$.\label{pty:ds:glue:planar:pw}
	\item $\optds(G_{\DS,1} \oplus \ett(I_1)) - \optds(G_{\DS,2} \oplus \ett(I_1)) \neq \optds(G_{\DS,1} \oplus \ett(I_2)) - \optds(G_{\DS,2} \oplus \ett(I_2))$.\label{pty:ds:nonequiv}
\end{enumerate}
\end{lemma}
\begin{proof}
Fix a boundary size~$t$ and let~$\mathcal{G}_{\IS}$ be the set of $t$-boundaried planar graphs constructed in Lemma~\ref{lemma:nonequivalent:planar}, which are pairwise nonequivalent under~$\equiv_{\IS,t}$. Observe that the construction of Lemma~\ref{lemma:planargraph} ensures that in all graphs in~$\mathcal{G}_\IS$, the boundary forms an independent set. Let~$\mathcal{G}_{\DS} := \{ \ett(G) \mid G \in \mathcal{G}_{\IS}\}$. By Observation~\ref{obs:ett:planar}, all graphs in~$\mathcal{G}_{\DS}$ are planar. Let~$G_{\DS,1}$ and~$G_{\DS,2}$ be distinct $t$-boundaried graphs in~$\mathcal{G}_{\DS}$. By definition of~$\mathcal{G}_{\DS}$, there exist distinct~$G_{\IS,1}$ and~$G_{\IS,2} \in \mathcal{G}_\IS$ such that~$G_{\DS,i} = \ett(G_{\IS,i})$ for~$i \in \{1,2\}$. By Property~\ref{noneqv:planargraphs} of Lemma~\ref{lemma:nonequivalent:planar}, there exist indicator graphs~$I_1$ and~$I_2$ such that
\begin{equation} \label{eq:is:notequivalent}
\optis(G_{\IS,1} \oplus I_1) - \optis(G_{\IS,2} \oplus I_1) \neq \optis(G_{\IS,1} \oplus I_2) - \optis(G_{\IS,2} \oplus I_2).
\end{equation}
We show that~$G_{\DS_1}$ and~$G_{\DS_2}$ satisfy the two claimed conditions with~$I_1$ and~$I_2$.

\textbf{(\ref{pty:ds:glue:planar:pw})} To see that the graphs~$G_{\DS,i} \oplus \ett(I_j)$ are planar for all choices of~$i$ and~$j$, observe that the boundaries of the two graphs being glued are independent sets. By Observation~\ref{obs:glue:ett} we therefore have~$G_{\DS,i} \oplus \ett(I_j) = \ett(G_{\IS,i}) \oplus \ett(I_j) = \ett(G_{\IS,i} \oplus I_j)$, which is planar by Observation~\ref{obs:ett:planar} since Lemma~\ref{lemma:nonequivalent:planar} guarantees that~$G_{\IS,i} \oplus I_j$ is planar. The same lemma guarantees that~$\pw(G_{\IS,i} \oplus I_j) = \pw(G_{\IS,i})$. Since the~$\ett$ operation increases the pathwidth by at most one (Proposition~\ref{prop:ett:pathwidth}), it follows that~$\pw(G_{\DS,i} \oplus \ett(I_j)) = \pw(\ett(G_{\IS,i} \oplus I_j)) \leq \pw(G_{\IS,i}) + 1 \leq t + \Oh(1)$, where the last step follows from the guarantee of Lemma~\ref{lemma:nonequivalent:planar}.

\textbf{(\ref{pty:ds:nonequiv})} Let~$n_i$ denote the number of vertices in~$G_{\IS,i}$ for~$i \in \{1,2\}$, and let~$m_j$ denote the number of vertices in~$I_j$ for~$j \in \{1,2\}$. Since two boundaries of size~$t$ are identified into a single size-$t$ boundary when gluing, we have~$|V(G_{\IS, i} \oplus I_j)| = n_i + m_j - t$ for all~$i,j \in \{1,2\}$. Because the complement of a maximum independent set is a minimum vertex cover, we have~$\optis(G_{\IS,i} \oplus I_j) = |V(G_{\IS, i} \oplus I_j)| - \optvc(G_{\IS, i} \oplus I_j) = n_i + m_j - t - \optvc(G_{\IS, i} \oplus I_j)$ for all~$i,j \in \{1,2\}$. Substituting these expressions into Equation~\ref{eq:is:notequivalent}, we get:
\begin{align*} \label{eq:is:notequivalent}
& (n_1 + m_1 - t - \optvc(G_{\IS,1} \oplus I_1)) - (n_2 + m_1 - t - \optvc(G_{\IS,2} \oplus I_1)) \\
\neq &(n_1 + m_2 - t - \optvc(G_{\IS,1} \oplus I_2)) - (n_2 + m_2 - t - \optvc(G_{\IS,2} \oplus I_2))
\end{align*}
After canceling terms, this simplifies to:
\begin{equation} \label{eq:vc:notequivalent}
\optvc(G_{\IS,1} \oplus I_1) - \optvc(G_{\IS,2} \oplus I_1) \neq \optvc(G_{\IS,1} \oplus I_2) - \optvc(G_{\IS,2} \oplus I_2)
\end{equation}
By Observation~\ref{obs:glue:ett}, we have~$\ett(G_{\IS,i} \oplus I_j) = \ett(G_{\IS,i}) \oplus \ett(I_j)$ for~$i,j \in \{1,2\}$, since the boundary forms an independent set in all graphs in~$\mathcal{G}_{\IS}$ and in all indicator graphs. We claim that the graphs~$G_{\IS,i} \oplus I_j$ for~$i,j \in \{1,2\}$ do not contain isolated vertices. To see this, note that the only vertices that can potentially be isolated in an indicator graph are its boundary vertices. The graphs~$G_{\IS,i}$ originate from the construction of Lemma~\ref{lemma:planargraph}, which is easily seen not to produce isolated vertices. Gluing a graph~$G_{\IS,i}$ onto an indicator graph therefore increases the degree of all boundary vertices to one or more, ensuring that~$G_{\IS,i} \oplus I_j$ does not have isolated vertices. We may therefore invoke Proposition~\ref{prop:vc:to:ds} on such graphs:
\begin{align*}
\optvc(G_{\IS,i} \oplus I_j) &= \optds(\ett(G_{\IS,i} \oplus I_j)) & \mbox{Proposition~\ref{prop:vc:to:ds}.} \\
&= \optds(\ett(G_{\IS,i}) \oplus \ett(I_j)) & \mbox{Observation~\ref{obs:glue:ett}.} \\
&= \optds(G_{\DS,i} \oplus \ett(I_j)) & \mbox{Definition of~$G_{\DS,i}$.}
\end{align*}
Substituting these expressions for~$\optvc$ into Equation~\ref{eq:vc:notequivalent}, we obtain:
\begin{align*}
& \optds(G_{\DS,1} \oplus \ett(I_1)) - \optds(G_{\DS,2} \oplus \ett(I_1)) \\
\neq& \optds(G_{\DS,1} \oplus \ett(I_2)) - \optds(G_{\DS,2} \oplus \ett(I_2))
\end{align*}
This shows that~$\mathcal{G}_\DS$ has all claimed properties and concludes the proof of Lemma~\ref{lemma:nonequiv:ds}.
\end{proof}

The counting argument of Theorem~\ref{theorem:planar:protrusion:lowerbound} that turns the lower bound on the number of equivalence classes of~$\equiv_\IS$ into a lower bound on the critical size of sets of planar representatives can be used without modifications to establish the same lower bound for \DominatingSet.

\begin{corollary}
Let~$t \geq t_0$ be a sufficiently large positive integer. Let~$\mathcal{R}_t$ be a set of $t$-boundaried planar graphs such that every equivalence class of~$\equiv_{\DS,t}$ that contains a planar graph of pathwidth~$t + \Oh(1)$ is represented by some graph in~$\mathcal{R}_t$. Then~$\mathcal{R}_t$ contains a graph with~$\Omega(\log M(t)) \geq \Omega(2^t / \sqrt{4t})$ vertices.
\end{corollary}

To conclude the section, we comment on how the remainder of the theory we developed for~$\equiv_\IS$ can be mirrored for~$\equiv_\DS$. For \IndependentSet, there are only two natural states for a boundary vertex: it may be included in the independent set, or it may not. For \DominatingSet, there are three natural states per vertex: (1) it is included in the dominating set, (2) it is not included but has to be dominated from the current $t$-boundaried graph, or (3) it is not included in the dominating set and will be dominated from the graph that is glued onto the current graph. Based on this difference, we expect an upper bound of~$2^{3^t}$ on the number of equivalence classes for~$\equiv_\DS$, compared to roughly~$2^{2^t}$ for~$\equiv_\IS$. We leave further \DominatingSet-analogues of our results, such as a characterization of the equivalence classes of~$\equiv_\DS$ by restricted monotone functions from~$\{0,1,2\}^t$ to~$\{0,1\}$, for future work.

\section{Conclusion} \label{sec:conclusion}
We presented lower and upper bounds on the number of equivalence classes of the canonical equivalence relation~$\equiv_{\IS,t}$ for \IndependentSet on $t$-boundaried graphs. We combined these lower bounds with upper bounds on the number of small graphs to give lower bounds for the critical sizes of sets of representatives for \IndependentSet and \DominatingSet. For a set of \emph{planar} representatives that represent all equivalence classes containing a bounded-pathwidth planar graph, we gave a lower bound of~$\Omega(\log M(t)) \geq \Omega(2^t / \sqrt{4t})$ on the critical size. The same argumentation can also be used to obtain lower bounds on the critical size of sets of potentially \emph{nonplanar} representatives. The number of distinct $t$-boundaried (unrestrained) graphs is at most~$2^{\binom{n}{2}} \cdot \binom{n}{t} \leq 2^{n^2 / 2}$. Using this bound in the proof of Theorem~\ref{theorem:planar:protrusion:lowerbound} yields a lower bound of~$\Omega(\sqrt{\log M(t)}) \geq \Omega(2^{t/2}/ \sqrt[4]{4t})$ on the critical size of a set of representatives that contains at least~$M(t) - 2$ distinct graphs.

In their work, Garnero et al.~\cite{GarneroPST15} (roughly) show that each equivalence class of~$\equiv_{\IS,t}$ containing a planar graph of treewidth at most~$t$ can be represented by a planar graph with~$2^{(t+1)^{2^t}}$ vertices and treewidth at most~$t$. Our lower bound shows that to represent all equivalence classes containing a planar graph of \emph{pathwidth}~$t + \Oh(1)$ (a subset of the graphs of treewidth~$t + \Oh(1)$), requires a graph with~$\Omega(2^t / \sqrt{4t})$ vertices. Our single-exponential lower bound is very far from the triple-exponential upper bound. However, we believe that the correct bound is single-exponential. Since Corollary~\ref{cor:characteristics} shows that each equivalence class is completely characterized by its normalized boundary function, and the construction of Lemma~\ref{lemma:generalgraphs} produces a boundaried graph with~$2^{\Oh(t)}$ vertices for any given boundary function, it follows that every equivalence class of~$\equiv_{\IS,t}$ has a representative with~$2^{\Oh(t)}$ vertices. Note, however, that the representatives constructed in this way are nonplanar and have pathwidth and treewidth~$2^{\Theta(t)}$.

The main conceptual contribution of this work is the fact that nontrivial lower bounds can be obtained by counting equivalence classes. The fact that a significant portion of the equivalence classes (at least~$M(t) \geq 2^{2^t / \sqrt{4t}}$ out of the total of at most~$2^{2^t}$) can be generated from monotone Boolean functions was useful in the construction of nonequivalent planar graphs of bounded pathwidth. 
We showed that the lower bound construction of Lokshtanov et al.~\cite{LokshtanovMS11a} can be planarized while increasing the pathwidth by an additive constant. The planarization argument employed here can also be used to strengthen the SETH-based runtime lower bound of~$\Omega((2-\varepsilon)^w \cdot n^{\Oh(1)})$ for solving \IndependentSet on graphs of treewidth~$w$, to \emph{planar} graphs of treewidth~$w$. Not all bounded-pathwidth graphs can be planarized with a bounded increase in pathwidth. In particular, when planarizing~$K_{3,n}$ for sufficiently large~$n$ the pathwidth grows arbitrarily large~\cite{Eppstein16}. 

\subparagraph*{Acknowledgments.} We are grateful to Daniel Lokshtanov and David Eppstein for insightful discussions regarding planarization, and to an anonymous referee of IPEC 2016 for suggesting a simplification in the proof of Theorem~\ref{theorem:planar:protrusion:lowerbound}.

\bibliography{references}

\begin{thebibliography}{10}

\bibitem{BodlaenderFLPST09}
Hans~L. Bodlaender, Fedor~V. Fomin, Daniel Lokshtanov, Eelko Penninkx, Saket
  Saurabh, and Dimitrios~M. Thilikos.
\newblock ({M}eta) {K}ernelization.
\newblock In {\em Proc. 50th FOCS}, pages 629--638. {IEEE} Computer Society,
  2009.
\newblock \href {http://dx.doi.org/10.1109/FOCS.2009.46}
  {\path{doi:10.1109/FOCS.2009.46}}.

\bibitem{BodlaenderFLPST13}
Hans~L. Bodlaender, Fedor~V. Fomin, Daniel Lokshtanov, Eelko Penninkx, Saket
  Saurabh, and Dimitrios~M. Thilikos.
\newblock ({M}eta) {K}ernelization.
\newblock {\em CoRR}, 2013.
\newblock \href {http://arxiv.org/abs/0904.0727} {\path{arXiv:0904.0727}}.

\bibitem{BodlaenderF01}
Hans~L. Bodlaender and Babette van Antwerpen{-}de~Fluiter.
\newblock Reduction algorithms for graphs of small treewidth.
\newblock {\em Inf. Comput.}, 167(2):86--119, 2001.
\newblock \href {http://dx.doi.org/10.1006/inco.2000.2958}
  {\path{doi:10.1006/inco.2000.2958}}.

\bibitem{BonichonGHPS06}
Nicolas Bonichon, Cyril Gavoille, Nicolas Hanusse, Dominique Poulalhon, and
  Gilles Schaeffer.
\newblock Planar graphs, via well-orderly maps and trees.
\newblock {\em Graphs and Combinatorics}, 22(2):185--202, 2006.
\newblock \href {http://dx.doi.org/10.1007/s00373-006-0647-2}
  {\path{doi:10.1007/s00373-006-0647-2}}.

\bibitem{CyganFKLMPPS15}
Marek Cygan, Fedor~V. Fomin, Lukasz Kowalik, Daniel Lokshtanov, D{\'{a}}niel
  Marx, Marcin Pilipczuk, Michal Pilipczuk, and Saket Saurabh.
\newblock {\em Parameterized Algorithms}.
\newblock Springer, 2015.
\newblock \href {http://dx.doi.org/10.1007/978-3-319-21275-3}
  {\path{doi:10.1007/978-3-319-21275-3}}.

\bibitem{Fluiter97}
Babette de~Fluiter.
\newblock {\em Algorithms for Graphs of Small Treewidth}.
\newblock PhD thesis, Utrecht University, 1997.

\bibitem{Eppstein16}
David Eppstein.
\newblock {P}athwidth of planarized drawing of~{$K_{3,n}$}.
\newblock {TheoryCS} {S}tack{E}xchange question, 2016.
\newblock URL: \url{http://cstheory.stackexchange.com/questions/35974/}.

\bibitem{FominLMPS16}
Fedor~V. Fomin, Daniel Lokshtanov, Neeldhara Misra, Geevarghese Philip, and
  Saket Saurabh.
\newblock Hitting forbidden minors: Approximation and kernelization.
\newblock {\em {SIAM} J. Discrete Math.}, 30(1):383--410, 2016.
\newblock \href {http://dx.doi.org/10.1137/140997889}
  {\path{doi:10.1137/140997889}}.

\bibitem{FominLMS12}
Fedor~V. Fomin, Daniel Lokshtanov, Neeldhara Misra, and Saket Saurabh.
\newblock Planar {$\mathcal{F}$}-deletion: Approximation, kernelization and
  optimal {FPT} algorithms.
\newblock In {\em Proc. 53rd FOCS}, pages 470--479. {IEEE} Computer Society,
  2012.
\newblock \href {http://dx.doi.org/10.1109/FOCS.2012.62}
  {\path{doi:10.1109/FOCS.2012.62}}.

\bibitem{FominS16}
Fedor~V. Fomin and Torstein J.~F. Str{\o}mme.
\newblock Vertex cover structural parameterization revisited.
\newblock {\em CoRR}, 2016.
\newblock \href {http://arxiv.org/abs/1603.00770} {\path{arXiv:1603.00770}}.

\bibitem{GajarskyHOORRVS13}
Jakub Gajarsk{\'{y}}, Petr Hlinen{\'{y}}, Jan Obdrz{\'{a}}lek, Sebastian
  Ordyniak, Felix Reidl, Peter Rossmanith, Fernando~Sanchez Villaamil, and
  Somnath Sikdar.
\newblock Kernelization using structural parameters on sparse graph classes.
\newblock In {\em Proc. 21st ESA}, pages 529--540. Springer, 2013.
\newblock \href {http://dx.doi.org/10.1007/978-3-642-40450-4_45}
  {\path{doi:10.1007/978-3-642-40450-4_45}}.

\bibitem{GareyJS76}
M.~R. Garey, David~S. Johnson, and Larry~J. Stockmeyer.
\newblock Some simplified {NP}-complete graph problems.
\newblock {\em Theor. Comput. Sci.}, 1(3):237--267, 1976.
\newblock \href {http://dx.doi.org/10.1016/0304-3975(76)90059-1}
  {\path{doi:10.1016/0304-3975(76)90059-1}}.

\bibitem{GarneroPST15}
Valentin Garnero, Christophe Paul, Ignasi Sau, and Dimitrios~M. Thilikos.
\newblock Explicit linear kernels via dynamic programming.
\newblock {\em {SIAM} J. Discrete Math.}, 29(4):1864--1894, 2015.
\newblock \href {http://dx.doi.org/10.1137/140968975}
  {\path{doi:10.1137/140968975}}.

\bibitem{Jansen13b}
Bart M.~P. Jansen.
\newblock {\em The Power of Data Reduction: Kernels for Fundamental Graph
  Problems}.
\newblock PhD thesis, Utrecht University, The Netherlands, 2013.
\newblock URL:
  \url{http://igitur-archive.library.uu.nl/dissertations/2013-0612-200803/UUindex.html}.

\bibitem{KimLPRRSS16}
Eun~Jung Kim, Alexander Langer, Christophe Paul, Felix Reidl, Peter Rossmanith,
  Ignasi Sau, and Somnath Sikdar.
\newblock Linear kernels and single-exponential algorithms via protrusion
  decompositions.
\newblock {\em {ACM} Trans. Algorithms}, 12(2):21, 2016.
\newblock \href {http://dx.doi.org/10.1145/2797140}
  {\path{doi:10.1145/2797140}}.

\bibitem{LokshtanovMS11a}
Daniel Lokshtanov, D{\'{a}}niel Marx, and Saket Saurabh.
\newblock Known algorithms on graphs on bounded treewidth are probably optimal.
\newblock In {\em Proc. 22nd SODA}, pages 777--789. {SIAM}, 2011.
\newblock \href {http://dx.doi.org/10.1137/1.9781611973082.61}
  {\path{doi:10.1137/1.9781611973082.61}}.

\bibitem{takahashi1995mixed}
Atsushi Takahashi, Shuichi Ueno, and Yoji Kajitani.
\newblock Mixed searching and proper-path-width.
\newblock {\em Theoretical Computer Science}, 137(2):253--268, 1995.
\newblock \href {http://dx.doi.org/10.1016/0304-3975(94)00160-K}
  {\path{doi:10.1016/0304-3975(94)00160-K}}.

\end{thebibliography}

\appendix

\end{document}